\documentclass[bibliography=totoc]{scrartcl}
\pdfoutput=1
\usepackage{scrlayer-scrpage}
\usepackage[utf8]{inputenc}
\usepackage{microtype}
\usepackage[ngerman, english]{babel}
\usepackage[german=quotes]{csquotes} 
\usepackage[bitstream-charter]{mathdesign}
\usepackage[T1]{fontenc}
\usepackage[EULERGREEK]{sansmath}
\usepackage{helvet}

\usepackage{xcolor}
\usepackage{amsmath}
\usepackage{amsthm}
\newtheorem{theorem}{Theorem}[section]
\newtheorem{lemma}{Lemma}[section]
\theoremstyle{remark} 
\newtheorem{remark}{Remark}[section]
\theoremstyle{definition} 
\newtheorem{definition}{Definition}[section]

\usepackage[style=alphabetic,sorting=nyt,hyperref,labelalpha,maxnames=100,mincitenames=4,maxcitenames=5]{biblatex}

\usepackage[hidelinks]{hyperref}
\usepackage{mathtools}

\makeatletter
\newcases{lrdcases}
{\quad}
{$\m@th\displaystyle{##}$\hfil}
{$\m@th\displaystyle{##}$\hfil}
{\lbrace}
{\rbrace}
\newcases{lrdcases*}
{\quad}
{$\m@th\displaystyle{##}$\hfil}
{{##}\hfil}
{\lbrace}
{\rbrace}
\makeatother

\usepackage[nameinlink,capitalize]{cleveref}
\usepackage{xspace}

\newcommand{\coloneq}{\coloneqq}
\newcommand{\eqcolon}{\eqqcolon}

\DeclareMathOperator*{\reverseProduct}{\mathchoice
	{\widetilde{\prod}}
	{\smash{\widetilde{\prod}}\vphantom{\prod}}
	{\widetilde{\prod}}
	{\widetilde{\prod}}
}
\newcommand{\smallReverseProduct}{\ensuremath{\widetilde\Pi}}
\newcommand{\diag}{\mathop{\mathrm{diag}}\nolimits}

\newcommand{\idxPlant}{_{\mathrm{p}}}

\newcommand{\subsDiscrete}{\mathrm{d}}
\newcommand{\idxDiscrete}{_{\subsDiscrete}}
\newcommand{\idxContinuous}{_{\mathrm{cont}}}

\newcommand{\subsActuate}{\mathrm{u}}
\newcommand{\subsMeasure}{\mathrm{y}}
\newcommand{\subsCompute}{\mathrm{ctrl}}

\newcommand{\todoOptional}[1]{}

\newcommand{\mat}[1]{\left\lbrack\begin{matrix} #1 \end{matrix}\right\rbrack}

\newcommand{\vecSmallT}[1]{\lbrack\begin{matrix} #1 \end{matrix}\rbrack^{\transp}}
\newcommand{\vecSubscriptT}[1]{\lbrack\begin{smallmatrix} #1 \end{smallmatrix}\rbrack^{\transp}}
\newcommand{\transp}{\top}

\newcommand{\R}{\mathbb{R}}
\newcommand{\Z}{\mathbb{Z}}
\newcommand{\N}{\mathbb{N}}
\newcommand{\ee}{\mathrm{e}}

\newcommand{\landauO}{O}

\newcommand{\ie}{i.\,e.\xspace}
\newcommand{\eg}{e.\,g.\xspace}

\usepackage{acronym}

\newcommand{\ifpaper}[1]{} 
\usepackage{changebar}

\newcommand{\ifreport}[1]{\cbstart#1\cbend} 
\newcommand{\ifreportX}[1]{#1} 
\newcommand{\ifreportStartMark}{\ifreportX{\cbstart}} 
\newcommand{\ifreportEndMark}{\ifreportX{\cbend}} 

\newcommand{\dDiff}{\mathrm{d}} 

\newcommand{\fracIfReport}[2]{\ifreportX{\frac{#1}{#2}}\ifpaper{#1/#2}} 
\newcommand{\spectralRadius}[1]{\rho\ifreportX{\hspace{-.08em}}\{#1\}}
\let\citep\cite

\newenvironment{keyword}{\noindent \textbf{Keywords: }}{}
\hypersetup{
pdftitle 	={Details and Proofs for: Stability Analysis of Multivariable Control Systems with Uncertain Timing},
pdfauthor 	={Maximilian Gaukler, Günter Roppenecker, Peter Ulbrich}
}

\date{April 24, 2020}
\begin{document}

\title{\ifreportX{Details and Proofs for: \\}
	Stability Analysis of Multivariable Digital Control Systems with Uncertain Timing
} 



\newcommand{\autoren}{Maximilian Gaukler, Günter Roppenecker, Peter Ulbrich}
\ifpaper{
	\author[First]{Maximilian Gaukler} 
	\author[First]{Günter Roppenecker} 
	\author[First]{Peter Ulbrich}
}

\newcommand{\affiliation}{Friedrich-Alexander-Universität Erlangen-Nürnberg (FAU), \\ 		
	\ifreportX{\small} Erlangen, Germany \ifreportX{\\ \small} (e-mail:~max.gaukler@fau.de, guenter.roppenecker@fau.de, peter.ulbrich@fau.de)}
\ifpaper{
	\address[First]{\affiliation}
}
\ifreportX{\author{\autoren\\ \small \affiliation}}

\ifreportX{\maketitle}

\begin{abstract}                
The ever increasing complexity of real-time control systems results in significant deviations in the timing of sensing and actuation, which may lead to degraded performance or even instability.
In this paper we present a method to analyze stability under \emph{mostly-periodic} timing with bounded uncertainty, a timing model typical for the implementation of controllers that were actually designed for strictly periodic execution.
In contrast to existing work, we include the case of multiple sensors and actuators with \emph{individual} timing uncertainty.
Our approach is based on the discretization of a linear impulsive system. To avoid the curse of dimensionality, we apply a decomposition that breaks down the complex timing dependency into the effects of individual sensor-actuator pairs. Finally, we verify stability by norm bounding and a Common Quadratic Lyapunov Function. Experimental results substantiate the effectiveness of our approach for moderately complex systems.
\end{abstract}

\begin{keyword}
Control over Networks, Systems with Time-Delays, Sampled-Data Systems, Jitter, \acl{CQLF}, \acl{LMI}, Linear Impulsive Systems
\end{keyword}

\ifpaper{\allowdisplaybreaks}
\newcommand{\specificReferenceToReport}[1]{\cite[\cref{REPORT-#1}]{Gaukler2019extended}}

\acrodef{LMI}{Linear Matrix Inequality}
\acrodefplural{LMI}{Linear Matrix Inequalities}
\acrodef{MIMO}{multiple-input-multiple-output}
\acrodef{CQLF}{Common Quadratic Lyapunov Function}
\acrodef{LIS}{Linear Impulsive System}
%
\todoOptional{define nondeterministic --- depending on the definition, hybrid automata which exploit may-semantics can still be deterministic, as long as the \emph{discrete} part of the automaton successor of a transition is unique. Here, nondeterministic = there is no deterministic input-output mapping.}

\ifreport{
\section*{Note}	
This document is the extended version of the following publication: \emph{Stability Analysis of Multivariable Digital Control Systems with Uncertain Timing -- M. Gaukler, G. Roppenecker and P. Ulbrich, 2020, IFAC World Congress.} 
Please cite the original publication wherever possible. 


Additions to the original publication are marked by a gray bar at the margin, as in this paragraph. While section numbers are equal as far as possible, the numbering of theorems and formulas differs from the original publication.
Coypright 2020 IFAC.
}

\ifreportX{\clearpage \tableofcontents\clearpage}

\section{Introduction}
The vast majority of control systems are implemented as discrete-time controllers executed on a real-time computing platform. In the design process, sampling the sensors and updating the actuators is generally assumed to be synchronous and strictly periodic. However, on modern computing platforms and due to the ever-growing overall system complexity, it is becoming increasingly difficult and often prohibitively costly to satisfy this assumption in the actual implementation:
First, execution times are non-constant and hard to predict, especially when multiple applications share one processor. Second, contemporary digital sensors incorporate excessive signal pre-processing.
Consequently, the sensor reading may be outdated by a small but varying duration, even if it is queried strictly periodically. Last but not least, the accuracy of time synchronization in distributed (\ie, ranging from multi-core to networked) systems is limited. All these factors jeopardize the controller's design assumptions and add to timing uncertainties in its input and output.

Therefore, the practical implementation of a controller with period $T$ will in most cases result in a \emph{mostly-periodic} system in which the sensor and actuator times do not lie on the intended periodic grid $t=kT$, $k \in \N$, but in a small \emph{timing window} around these points. The resulting dynamics may be worse or even unstable. In practice, it is often assumed that the timing window is still small enough such that stability and convergence are not affected. This argument is problematic for two reasons: Firstly, without proper analysis, there is no guarantee that a certain timespan is \enquote{small enough}. Secondly, larger timing windows relax and simplify the scheduling of real-time applications and are therefore even desirable from a (real-time) design point of view. Consequently, in this paper, we concentrate on the stability analysis of mostly-periodic digital control loops with given timing windows.

\section{Related Work}

\acrodef{LET}{logical execution time}
Providing a deterministic execution platform has always been a core aim in real-time scheduling and design. Here, the general approach to eliminate timing uncertainty is to rely on a time-triggered execution of the controller code at predetermined instants of time. Known representatives for this are the Cyclic Executive~\citep{baker:89:rts} and Fixed-Priority Models~\citep{sha:89:real} for periodic tasks. However, the focus is on deadline adherence rather than avoidance of jitter. Synchronous development models address the latter problem. For example, the \ac{LET} paradigm \citep{Henzinger2003} suggests a decomposition of input and output: Sensors are sampled at fixed time instants (\eg, $t=kT$). Instead of updating the output immediately after the new value has been computed, the update is delayed until $t=kT+D_{\subsActuate}$ to eliminate jitter. In general, support for \emph{exact} synchronization requires, however, tailored programming languages and hardware support and is thus inapplicable to a wide range of systems. Therefore, most practical implementations of \ac{LET} resort to overapproximations and pessimism to match synchronicity within some uncertainty, which results in a timing window as considered in this work.

\acrodef{SISO}{single-input-single-output}
For the analysis of sampled-data systems with uncertain timing, a wide array of theoretical methods is available (cf. \cite{hetel2017stability}). From a user's point of view, the existing results building upon these methods can be categorized by the employed timing model:

Based on the small gain theorem, \cite{cervin12acc} analyzes stability for a timing model similar to ours. The analysis is, however, restricted to the \ac{SISO} case, which is easier since there are only two scalar timing uncertainties, namely sensor and actuator delay. The same holds for multiple inputs and outputs if all sensors are jointly sampled and all actuators are jointly updated. This results in a system with \ac{SISO}-like timing but vector-valued signals (\enquote{quasi-\ac{SISO}}). However, the quasi-SISO assumption is invalid for systems with multiple sensors that are not exactly synchronized.

Quasi-\ac{SISO} cases are analyzed in \cite{Kao2007,khatib:16:jnahs,Bauer2012ncs} and, with restriction to quantized output delays, in \cite{fontanelli:13:automatica}. To model network-controlled systems, \cite{Bauer2012ncs} also offers the alternative model that exactly one sensor or actuator is updated in every control period, thereby transforming a \ac{MIMO} system to a switched quasi-\ac{SISO} one. As this scenario is tailored to networked control with severely restricted communication resources, it does not match the common scenario of an embedded system that has enough resources to query all sensors in every period.

To the best of our knowledge, none of the existing publications address the actual \ac{MIMO} case of multiple sensors and actuators with independent timing uncertainties. Filling this critical gap is the contribution of this work.

\section{Problem Statement}\label{sec:problem}
\paragraph*{System Model:}  A control loop that is exponentially stable for perfect timing is executed with uncertain timing. We employ the system model by \cite{Gaukler2019}, restricted to the linear case without disturbance and measurement uncertainty:

The plant $\dot x\idxPlant(t) = A\idxPlant x\idxPlant(t) + B\idxPlant u(t)$ with state $x\idxPlant(t) \in \R^{n\idxPlant}$, output $y(t) = C\idxPlant x\idxPlant(t) \in \R^p$ and input $u(t) \in \R^m$ is controlled by a discrete-time controller with fixed period $T>0$, state $x\idxDiscrete(t) \in \R^{n\idxDiscrete}$. The controller dynamics are
\begin{equation}
	y_{\subsDiscrete,k} = y_k, ~~~ x_{\subsDiscrete,k+1} = A\idxDiscrete x_{\subsDiscrete,k} + B\idxDiscrete y_{\subsDiscrete,k}, ~~~ u_k = C\idxDiscrete x_{\subsDiscrete,k},
\end{equation}
where $y\idxDiscrete \in \R^p$ is a measurement buffer introduced for formal reasons.

Under ideal timing, the measurement $y_{\subsDiscrete,k}$ and actuation $u$ would be updated at $t=kT$. Actually, updating the $i$-th actuator component $u^{(i)}$ is offset by the timing deviation $\Delta t_{\subsActuate,i,k}$ and, respectively, sampling $y^{(i)}$ by $\Delta t_{\subsMeasure,i,k}$:
\begin{align}
	u^{(i)}(t) &= u_k~~\text{for}~~kT + \Delta t_{\subsActuate,i,k} \le t < (k+1)T + \Delta t_{\subsActuate,i,k+1},\nonumber\\
	y_{\subsDiscrete,k}^{(i)} &= y^{(i)}(kT + \Delta t_{\subsMeasure,i,k}).
\end{align}
The timing deviations are unknown but bounded to
\begin{align}
	\underline{\Delta t}_{\{\subsActuate,\subsMeasure\},i} \le \Delta t_{\{\subsActuate,\subsMeasure\},i,k} \le \overline{\Delta t}_{\{\subsActuate,\subsMeasure\},i},\\
	\intertext{where the bounds are less than half a period:}
	-T/2 < \underline{\Delta t}_{\{\subsActuate,\subsMeasure\},i} \le \overline{\Delta t}_{\{\subsActuate,\subsMeasure\},i} < T/2.
\end{align}
\ifpaper{\vspace{-2em}}
\paragraph*{Formalization:} To achieve a uniform formulation, the \enquote{discrete-time} variables $u$, $x\idxDiscrete$ and $y\idxDiscrete$ are treated as continuous-time signals that are updated at certain times and remain constant inbetween (zero-order hold). In this formulation, the $k$-th control period ($k \in \N$) is executed as follows: At $t_{\subsMeasure,i,k} = kT + \Delta t_{\subsMeasure,i,k}$, the $i$-th sensor, $i=1,...,p$, is sampled by setting the $i$-th component of $y\idxDiscrete(t)$ to the $i$-th component of $y(t)$. Similarly, the $j$-th actuator, $j=1,...,m$, is updated at $t_{\subsActuate,j,k} = kT + \Delta t_{\subsActuate,j,k}$ by setting the $j$-th component of $u(t)$ to the $j$-th component of $C\idxDiscrete x\idxDiscrete(t)$. Finally, the discrete controller is updated at $t=(k+1/2)T$ by setting $x\idxDiscrete(t)=A\idxDiscrete x\idxDiscrete(t^-) + B\idxDiscrete y\idxDiscrete(t^-)$. As discussed later, fixing this time at $t=(k+1/2)T$ is without loss of generality; it may be earlier or later as long as the order of events is maintained.
 
For readability, the startup behavior is defined such that the $0$-th control period is skipped and the initial states are given at $t_0=T/2$.
The resulting system is linear but nondeterministic and time-variant. For a detailed discussion of this model, see \cite{Gaukler2019}.
\ifpaper{\vspace{-1.5em}}
\paragraph*{Goal:} We want to prove exponential stability of the closed loop for moderate timing uncertainties. The focus is on an efficient solution that scales well to systems with a large number of inputs and outputs, even if this scalability makes the result more pessimistic and therefore the approach is only applicable to small timing uncertainties.

We define stability as the exponential decay of plant state $x\idxPlant$, controller state $x\idxDiscrete$, sampled measurement $y\idxDiscrete$ and actuation $u$, which are combined in the state vector
\begin{equation}
x(t) = \mat{x\idxPlant(t)^\transp & x\idxDiscrete(t)^\transp & y\idxDiscrete(t)^\transp & u(t)^\transp}^\transp  \in \R^{n}
\end{equation}
of dimension $n=n\idxPlant + n\idxDiscrete + p + m$:

\acrodef{CGES}{Continuous-Time Globally Uniform Exponential Stability}
\begin{definition} The closed loop with initial state $x(t_0)$ admits \emph{\acl{CGES}}\acused{CGES}, denoted as \ac{CGES}$(\lambda, D)$, iff there exist constants $D \in [1,\infty)$ and $\lambda < 0$ such that for all possible timings
	\begin{equation}
	|x(t)| \leq D|x(t_0)| \ee^{\lambda (t-t_0)} \quad \forall t\geq t_0, ~\forall x(t_0) \in \R^n.
	\end{equation}
\end{definition}

\section{\ifreportX{Preliminaries and }Notation} \label{sec:timing-stability:preliminaries}

Definitions are denoted with a colon, \eg, $a \coloneq b$ means that $a$ is defined as $b$. We define $\R$ as the real numbers, $\N := \{1,2,\dots\}$ and $\Z:=\{0,\pm1,\dots\}$. For a set $S$, the number of elements is denoted $|S|$. Rounding down is $\lfloor x \rfloor \coloneq \max \{ z \in \Z ~|~ z \leq x \}$. The euclidean norm of $x \in \R^{n}$ is $|x|:=\sqrt{x^{\transp} x}$, where ${}^{\transp}$ denotes transposition. If $A\in \R^{n \times n}$ has eigenvalues $\lambda_i$, it has spectral radius $\spectralRadius{A} := \max_i |\lambda_i|$. The spectral norm is $\|A\|_{\sigma}:=\spectralRadius{A^{\transp}A}$.

For $a,b\in\Z$, the reversed product $\smallReverseProduct$ is defined as
\ifpaper{\vspace{-0.2em}}
\begin{equation}
\reverseProduct_{i=a}^b X_i \coloneq \prod_{\mathclap{i=-b}}^{-a} X_{-i} = \begin{cases}
X_{b} X_{b-1} \dots X_{a+1} X_{a},& a \leq b,\\
I,& a > b.
\end{cases} \label{eq:timing-stability:reverse-product-def}
\end{equation}
\ifpaper{\vspace{-0.9em}}
\ifreport{$I$ is the unity matrix and $e_j=\mat{0 & \dots & 0 & 1 & 0 & \dots & 0}^\transp$ the $j$-th standard basis vector, both of appropriate dimension.}

\ifpaper{Positive definiteness of functions and matrices is denoted by $f(x)\succ 0$ and $P\succ0$.}
\ifreport{
	\begin{definition}[Positive Definiteness] \label{def:positive-definite}
		For the functions $f,g: \R^n \mapsto \R$, we define
		\begin{align}
		f(x) &\succ g(x) \quad\ifpaper{\!\!\!\!} :\Leftrightarrow \quad\ifpaper{\!\!\!\!} \begin{lrdcases}f(x)>g(x),& x \neq 0\\ f(x)=g(x),& x=0 \end{lrdcases} \quad\ifpaper{\!\!\!\!} \forall x \in \R^n.
		\intertext{For the symmetric matrices $F=F^\transp,G=G^\transp\in \mathbb R^{n \times n}$,}
		F &\succ G \quad :\Leftrightarrow \quad x^{\transp} (F-G)\, x \succ 0.
		\end{align}
		
		\ifreport{To define positive semidefiniteness, negative definiteness and negative semidefiniteness ($\succeq, \prec, \preceq$), the relation $>$ is replaced by $\geq, <, \leq$ respectively.}
		\ifreport{\xspace
			The restriction to symmetric $F$ and $G$ simplifies the further derivations, but does not restrict the results because only the symmetric part of a matrix contributes to the quadratic form:
			\begin{equation}
			x^\transp M x = x^\transp \left(\frac{M + M^\transp}{2}\right) x \quad \forall x \in \R^n, M \in \R^{n \times n}.
			\end{equation}}
	\end{definition}
	
	\begin{definition}[Norm\ifreport{~ \cite[pp. 597, 601]{Bernstein2009}}]
		\label{def:norm}
		The function $\|\cdot\|$ : $S \mapsto [0,\infty)$ is a vector norm on $S=\R^n$ or matrix norm on $S=\R^{m \times n}$ iff $\forall x,y \in S,   \alpha \in \R$
		\ifpaper{\allowdisplaybreaks}
		\begin{subequations}
			\begin{align}
			\|x\| &
			\begin{cases}
			> 0, & x \neq 0,\\
			= 0, & x = 0
			\end{cases}
			\\
			\|\alpha x\| &= |\alpha| \|x\| \\
			\text{and} \quad 	\| x + y \| &\le \|x\| + \|y\|. \label{eq:timing-stability:norm-additive}
			\end{align}
		\end{subequations}
	\end{definition}
	
	\begin{definition}[Euclidean \ifreport{Vector }Norm]
		For $x \in \R^{n}$, $|x|\coloneq \sqrt{x^\transp x}$ denotes the Euclidean norm, which is a vector norm.
	\end{definition}
	
	\begin{theorem}[Equivalence of Norms] \label{theorem:timing-stability:norm-comparison}
		\ifpaper{~}All norms $\|\cdot\|_a$, $\|\cdot\|_b$ defined on the same space $S$ are equivalent up to a bounded factor:
		\begin{equation}
		\ifpaper{\renewcommand{\quad}{~}}
		\forall \|\cdot\|_a, \|\cdot\|_b \quad \exists c_1,c_2>0\ifpaper{\!}:\quad \forall x \ifpaper{~}\quad c_1\|x\|_b \le \|x\|_a \le c_2\|x\|_b.\ifpaper{\!}
		\end{equation}
		\cite[Theorem 9.1.8, Definition 9.2.1]{Bernstein2009}
	\end{theorem}
	
	\begin{definition}[Matrix-valued Limit]
		\label{def:timing-stability:limit}
		Analogous to the classical epsilon-delta-definition \citep{Stover2019Limit}, a matrix- or vector-valued limit is defined as
		\begin{align}
		\lim\limits_{x\to a} f(x) = y  \quad \ifpaper{\hspace{-.5em}} :\Leftrightarrow \quad \ifpaper{\hspace{-.5em}} \Big[\,&\forall \epsilon > 0 ~\exists \delta(\epsilon) > 0 \text{ such that }\quad \nonumber\\&\forall x \text{ with } \|x-a\|_X<\delta:\quad \ifpaper{\nonumber\\&\hspace{4.5em}} \|f(x) - y\|_Y < \epsilon\,\Big].
		\end{align}
		
		Due to \cref{theorem:timing-stability:norm-comparison}, the result is independent of the chosen norms $\|\cdot\|_X$ and $\|\cdot\|_Y$.
	\end{definition}
	
	\begin{definition}[Submultiplicative Matrix Norm\ifreport{~\cite[p. 604]{Bernstein2009}}]
		A matrix norm $\|\cdot\|$ is submultiplicative iff
		\begin{equation}
		\|XY\|\leq \|X\|\|Y\| \quad \forall X,Y \in \R^{n \times n}.
		\end{equation}
	\end{definition}
	
	\begin{definition}[Equi-Induced Matrix Norm\ifreport{~\cite[pp. 607 f.]{Bernstein2009}}]
		\label{def:timing-stability:equi-induced-norm}
		Every vector norm $\|\cdot\|_v$ on $\R^n$ leads to a corresponding \emph{equi-induced matrix norm} $\| \cdot \|_v'$ on $\R^{n \times n}$, defined by
		\begin{equation}
		\|M\|_v' \coloneq \max_{x \in \R^n \setminus \lbrace 0 \rbrace} \ifpaper{\!}\frac{\|Mx\|_v}{\|x\|_v}
		\ifreport{= \max_{\substack{x \in \R^n\text{ with }\|x\|_v = 1}} \|Mx\|_v},
		\end{equation}
		which is submultiplicative.\ifpaper{ \cite[pp. 607 f.]{Bernstein2009}} \ifreport{The prefix \enquote{equi-} denotes that $M$ is square.}
	\end{definition}
	
	\begin{definition}[Spectral Norm]
		The spectral norm
		\begin{align}
		\|M\|_\sigma :=&  \spectralRadius{A^{\transp}A} = \ifreportX{\bar \sigma(M) =}
		\max\limits_{x \in \R^n \setminus \{0\}} \frac{|Mx|}{|x|}
		\ifreportX{= \max\limits_{x~\text{with}~|x|=1} \underbrace{\sqrt{x^{\transp} M^{\transp} M x}}_{|Mx|}}
		\end{align}
		of $M$ is the maximum singular value $\bar\sigma(M)$, which describes the maximum growth of the euclidean norm \mbox{$|\cdot|$} due to multiplication with $M$. It is the equi-induced matrix norm of the euclidean vector norm, and therefore a submultiplicative matrix norm. 
		\cite[pp. 328, 
		603, 
		607--609
		]{Bernstein2009}%
	\end{definition}
	
	\begin{theorem}
		\label{theorem:timing-stability:norm-of-exp}
		If $\|\cdot\|$ is a submultiplicative matrix norm on $\R^{n\times n}$ and $\|I_{n \times n}\|=1$, then
		\begin{equation}
		\|\ee^{A\delta}\| \leq \ee^{\|A\delta\|} = \ee^{\|A\| |\delta|} \quad \forall A \in \R^{n \times n}, \delta > 0
		\end{equation}
		\cite[Proposition 11.1.2]{Bernstein2009}. The requirement $\|I_{n \times n}\|=1$ is fulfilled for all equi-induced norms \cite[Theorem 9.4.2]{Bernstein2009}, \ifpaper{including}\ifreport{which includes}\xspace the spectral norm.
	\end{theorem}
}
\ifpaper{The Cholesky Decomposition of $P \succ 0$ is $P =: P^{1/2} (P^{1/2})^{\transp}$. We use common definitions and properties of matrix norms from \cite{Bernstein2009}; a detailed list is detached to the extended version of this paper \specificReferenceToReport{sec:timing-stability:preliminaries}.}
\ifreport{
	\begin{theorem}[Cholesky Decomposition\ifreport{~\cite[Fact 8.9.38]{Bernstein2009}}] \label{theorem:timing-stability:cholesky}
		Any $P \succ 0$ \ifreport{$(\succeq 0)$\xspace} can be decomposed into
		$
		P =: P^{1/2} (P^{1/2})^{\transp}
		$ such that $P^{1/2}$ is lower triangular with positive \ifreport{(nonnegative) }diagonal entries\ifreport{. For $P \succ 0$, $P^{1/2}$ is}\ifpaper{,}\xspace invertible and uniquely defined. 
	\end{theorem}
}

\section{Approach}
This section presents the high-level structure of our approach. Details are given in the subsequent sections.

\emph{Discretization:} In \cref{sec:robust-stability:discretization}, we apply a time discretization
\ifpaper{\vspace*{-.5em}}
\begin{equation}
	x_k \coloneq x(t_k^+) \coloneq \lim \limits_{\epsilon \to 0^+} x(kT + T/2 +\epsilon),
\end{equation}
which leads to the linear discrete-time system $x_{k+1} = A_k x_k$, whose transition matrix $ A_k = A(\Delta t_k) $ depends on the current timing vector
\begin{equation}
\Delta t_k \coloneq \mat{\Delta t_{\subsActuate,1,k}& \dots& \Delta t_{\subsActuate,m,k}& \Delta t_{\subsMeasure,1,k}& \dots& \Delta t_{\subsMeasure,p,k}}^{\transp}.
\end{equation}
The offset $+T/2$ was chosen such that the sensing and actuation events cannot move across the discretization times. This ensures that $A_k$ depends only on $\Delta t_k$.

In the following, the subscript $k$ of timing variables $\Delta t_{\dots}$ is often omitted.
To further simplify the notation, the system dynamics are defined as right-side continuous, so that always $x(t^+) = x(t)$. Therefore, the discretization is simplified to
$x_k :=  x(t_k)$ with $t_k := kT + T/2$.

Stability of the discretized system is easier to analyze, but still \emph{equivalent} to the desired continuous-time stability:

\acrodef{DGES}{Discrete-Time Globally Uniform Exponential Stability}
\begin{definition}
	The discretized control loop
	\begin{equation}
	x_{k+1} = A_k x_k, \quad A_k \in \mathcal A \subset \R^{n \times n} \label{eq:jsr-uncertain-discrete-sys}
	\end{equation}
	admits \emph{\acl{DGES}}\acused{DGES}, denoted as \enquote{$\mathcal{A}$ is \ac{DGES}$(\rho, C)$}, iff there exist constants $C \in [1,\infty)$ and $\rho \in (0, 1)$ such that
	\begin{equation}
	|x_k| \leq C|x_0| \rho^{k} \quad \forall k\geq0, \forall x_0 \in \R^n, \forall A_{0}, A_{1}, ... \in \mathcal{A}.
	\end{equation}	
	\ifreport{(Note that the restrictions $C\geq1$ and $\rho \neq 0$ immediately follow from the above equation.) }
\end{definition}
Here,
$
\mathcal{A}=\{A(\Delta t_k) ~|~ \underline{\Delta t}_{\{\subsActuate,\subsMeasure\},i} < \Delta t_{\{\subsActuate,\subsMeasure\},i,k} < \overline{\Delta t}_{\{\subsActuate,\subsMeasure\},i}\}
$
is the set of possible $A_k$ for all possible timings $\Delta t_{\{\subsActuate,\subsMeasure\},i,k}$.
\begin{theorem}
	\label{theorem:timing-stability:cges-equiv-dges}
	For the given control loop, \ac{CGES} $\Leftrightarrow$ \ac{DGES}.
\end{theorem}

\begin{proof}
	The proof given in \cref{sec:timing-stability:discretization-proof} works by bounding the overshoot inbetween two discretization points. \ifpaper{\qed}
\end{proof}

Next, we want to show \ac{DGES} by a \ac{CQLF}: Find $P \in \R^{n \times n}$ such that
\begin{align*}
V_P(x):=x^{\transp} P x &\succ 0 \quad 
\text{with} \quad V_P(A_k x) \prec V_P(x) \quad \forall A_k \in \mathcal{A}.
\end{align*}
\ifpaper{\vspace{-2.5em}}
\paragraph*{Difficulty:} 
To the best of our knowledge, the straightforward extension of an existing method is not feasible:
	
A direct numerical approach based on a grid of possible $\Delta t_k$ (\eg grid-and-bound as in \cite{Heemels2010}) suffers from \emph{exponential complexity} with regard to the number $m+p$ of sensors and actuators, which is also the dimension of the timing parameter space.
	
Similarly, an analytical approach which directly uses an explicit expression for $A(\Delta t_k)$ suffers from the prohibitively large number of case distinctions corresponding to the $(m+p)!$ possible orderings of sensor and actuator times.
\ifpaper{\vspace{-2em}}
\paragraph*{Decomposition (\cref{sec:robust-stability:splitting}):} We avoid these difficulties by breaking up the dynamics into a sum:
\begin{theorem}[Decomposition]
	 The transition matrix $A$, which depends on $m+p$ scalar timing variables, can be split into a sum of functions of one scalar parameter each:
	 	\allowdisplaybreaks
	\begin{align}
	A(\Delta t) = &A(\Delta t = 0) +  \sum_{i=1}^{m} \Delta A_{\subsActuate,i} (\Delta t_{\subsActuate,i})  \ifpaper{\nonumber \\&} +
	\sum_{j=1}^{p} \Delta A_{\subsMeasure,j} (\Delta t_{\subsMeasure,j})  \nonumber \\& +
	\sum_{i=1}^{m} \sum_{j=1}^{p}  \Delta A_{\subsActuate\subsMeasure,i,j} (\Delta t_{\subsMeasure,j} - \Delta t_{\subsActuate,i}),
	\label{eq:timing-stability:splitting}
	\end{align}
	where $A(\Delta t=0)$ is the nominal case and $\Delta A_{\dots}$ are \enquote{deviations} that obey $\lim_{|\Delta t| \to 0} \Delta A_{\dots}=0$.
	\label{theorem:timing-stability:splitting}
\end{theorem}
\begin{proof}See  \cref{sec:robust-stability:splitting} for the proof and results. \ifpaper{\qed}
\end{proof}
Loosely interpreted, $\Delta A_{\subsActuate,i}$ is the deviation of $A$ resulting from the timing of the $i$-th actuator, $\Delta A_{\subsMeasure,j}$ corresponds to the $j$-th sensor, and $\Delta A_{\subsActuate\subsMeasure,i,j}$ to the influence of actuator $i$ on sensor $j$. Explicit expressions are given in \cref{sec:robust-stability:splitting}.

\ifpaper{\vspace{-1.5em}}
\paragraph*{Stability by Norm Bounding (\cref{sec:timing-stability:norm,sec:timing-stability:normbound,sec:lmi}):} As assumed in the problem setting, the nominal case (perfect timing $\Delta t=0$) is stable and therefore achieves \ac{DGES} with $\rho < 1$. The resulting safety margin $1-\rho > 0$ can be used to prove stability up to a certain amount of timing deviation. For this, we use a matrix norm corresponding to a \ac{CQLF}:

\begin{theorem} Let $V_P(x)=x^\transp P x$, $P \in \R^{n\times n}$, be a positive definite\ifreport{\xspace(Lyapunov candidate)}\xspace function. Then the \emph{$P$-ellipsoid norm}
	\begin{equation}
	\|A\|_P := \max_{x \neq 0} \sqrt{\frac{V_P(Ax)}{V_P(x)}}
	\end{equation}
	is a submultiplicative matrix norm.
	\label{theorem:timing-stability:p-norm}
\end{theorem}
\begin{proof}
	See \cref{sec:timing-stability:norm}. \ifpaper{\qed}
\end{proof}
This norm $\|A\|_P$ represents the worst-case decay of $V_P(x)$ for the time-invariant system $x_{k+1}=Ax_k$:
\begin{equation}
\|A\|_P \leq \rho \quad \Leftrightarrow \quad \left(V_P(x_{k+1}) \leq \rho^2 V_P(x_k) \quad\forall x_k\right).
\end{equation}

In general, norm bounds can be highly pessimistic. However, this norm can \emph{accurately} capture stability of the \emph{nominal} case $x_{k+1}=A(\Delta t=0)x_k$, for which $\spectralRadius{ A(\Delta t=0) } < 1$ is the minimal possible stability factor $\rho$ for \ac{DGES}.
\begin{theorem}
There exists $P$ such that $\rho_{\mathrm{n}} \coloneq \|A(\Delta t = 0)\|_P$ is arbitrarily close to $\spectralRadius{ A(\Delta t=0) }$.
	\label{theorem:timing-stability:qlf-to-norm-simplified}
\end{theorem}
\begin{proof}
	See \cref{sec:timing-stability:norm}, \cref{theorem:timing-stability:extreme-lyapunov-lti}. \ifpaper{\qed}
\end{proof}

To check stability for uncertain timing, choose any $P \succ 0$ for which $\rho_{\mathrm{n}}\ifpaper{\!}<1$. This exists by the previous theorem; the implementation is discussed later. Then, stability under uncertain timing can be shown if the summands $\Delta A_{\dots}$ in \eqref{eq:timing-stability:splitting}, which represent timing deviation, are small enough:
\begin{theorem}[Norm Bounding]
	\label{theorem:timing-stability:norm-bound}
	The system is \ac{DGES} if
	\ifpaper{\allowdisplaybreaks}
	\begin{align}
	\Bigg(\underbrace{\|A(\Delta t = 0)\|_P}_{= \rho_{\mathrm{n}}} +  \sum_{i=1}^{m} {\|\Delta A_{\subsActuate,i} (\Delta t_{\subsActuate,i})\|_P}
	\nonumber \\
	+
	  \sum_{j=1}^{p} {\|\Delta A_{\subsMeasure,j} (\Delta t_{\subsMeasure,j})\|_P}  	\nonumber \\
	+  {\sum_{i=1}^{m} \sum_{j=1}^{p}} {\|\Delta A_{\subsActuate\subsMeasure,i,j} ({\Delta t_{\subsMeasure,j} - \Delta t_{\subsActuate,i}})\|_P} \label{eq:timing-stability:norm-bound:sum} & \Bigg)  < 1 \nonumber\\
	~~~\forall  \Delta t_{\{\subsActuate,\subsMeasure\},i} \in \left(\underline{\Delta t}_{\{\subsActuate,\subsMeasure\},i};~ \overline{\Delta t}_{\{\subsActuate,\subsMeasure\},i} \right).\hspace{-5em}&\hspace{5em}
	\end{align}
\end{theorem}
\begin{proof}
	Consider $\|A(\Delta t)\|_P$ and apply \cref{theorem:timing-stability:splitting} and the triangle inequality to see that \eqref{eq:timing-stability:norm-bound:sum} implies $\|A(\Delta t)\|_P < 1$ for all possible $\Delta t$. This leads to \ac{DGES} as detailed in \cref{sec:timing-stability:norm}, \cref{theorem:timing-stability:norm-bound:generic}. \ifpaper{\qed}
\end{proof}

For the practical implementation, upper bounds for $\|\Delta A_{\dots}\|_P$ are computed in \cref{sec:timing-stability:normbound} and $P$ is determined by \acp{LMI} in \cref{sec:lmi}.

\paragraph*{Benefits:} The approach shows \ac{DGES} and therefore \ac{CGES} using norm bounds, which entails some conservatism. This will later be evaluated by experiments in \cref{sec:experiments}. On the other hand, the chosen method is particularly well-suited for analyzing \ac{MIMO} systems with moderate timing uncertainty:
\begin{theorem}[Stability implies timing robustness]
	\label{theorem:timing-stability:feasibility}
	If the nominal case $\Delta t=0$ is stable, then \cref{theorem:timing-stability:norm-bound} can show stability for some nonzero (possibly small) timing bounds.
\end{theorem}
\begin{proof}
	Consider the summands $\rho_{\mathrm{n}} + \sum \|\Delta A_{\dots}\|_P$ in \eqref{eq:timing-stability:norm-bound:sum}. Assume a timing bound $|\Delta t| < \delta$ with sufficiently small $\delta > 0$.
	\ifpaper{
		By \cref{theorem:timing-stability:qlf-to-norm-simplified}, choose $P$ such that $\rho_{\mathrm{n}} < 1$. By \cref{theorem:timing-stability:splitting}, $\Delta A_{\dots} \to 0$ for $|\Delta t| \to 0$, so choosing $\delta$ sufficiently small guarantees that $\sum \|\Delta A_{\dots}\|_P < 1-\rho_n$. Then, \eqref{eq:timing-stability:norm-bound:sum} is true. For a detailed proof, see \specificReferenceToReport{theorem:timing-stability:feasibility}.
		\qed
	}
	\ifreport{
		For the first summand, the assumption of nominal stability $\spectralRadius{A(\Delta t=0)} < 1$ means that due to \cref{theorem:timing-stability:qlf-to-norm-simplified}, it is possible to choose $P$ such that $\rho_{\mathrm{n}} < 1$.
		
		The next step is to bound the remaining part of the sum below $1-\rho_{\mathrm{n}}$.
		Due to \cref{theorem:timing-stability:splitting}, $\Delta A_{\dots} \to 0$ for $\Delta t \to 0$. By the definition of a matrix-valued limit (\cref{def:timing-stability:limit}), this implies that for any desired bound $\epsilon > 0$ on the norm $\|\Delta A_{\dots}\|_P$ of the deviations $\Delta A_{\dots}$ from the nominal case, there is a corresponding timing bound $\delta(\epsilon)>0$ such that $(|\Delta t| < \delta \Rightarrow \|\Delta A_{\dots}\|_P < \epsilon)$. Let $\epsilon$ be small enough such that the condition of \cref{theorem:timing-stability:norm-bound} is satisfied. Then, the system is stable for $|\Delta t| < \delta(\epsilon) > 0$.
	}
\end{proof}
	
\ifreport{
	Because nominal stability ($\rho_{\mathrm{n}}<1$) must hold for the result of any controller design method, this has two important consequences:
		
	\begin{itemize}
		\item 	In theory, the approach is always guaranteed to return some nonzero timing range. In practice, numerical issues of the implementation may prevent success if $\rho_{\mathrm{n}}$ is very close to $1$.
		
		\item Independent of the approach, any control loop of the considered form that is stable for perfect timing is also stable for a small amount of timing deviation, even if timing or robustness were not considered in the design.
	\end{itemize}
}

\begin{remark}[Complexity]
	\label{remark:timing-stability:complexity}
	With increasing number of sensors and actuators, checking \cref{theorem:timing-stability:norm-bound} requires only a \emph{polynomially} increasing number of matrix norm computations. The approach therefore avoids the exponential growth suffered by gridding the parameter space. In detail, the computation consists of determining $P$, $\rho_{\mathrm{n}}$, and then $p+m+mp$ bounds one-dimensional functions $\|\Delta A_{\dots} (\delta)\|_P$, where $\delta$ is a bounded scalar variable.
	\todoOptional{Komplexität weiter ins Detail aufdröseln -- Dimension der Matrizen wird ja auch größer...siehe sec:timing-stability:normbound:complexity in unused.tex.}
	\todoOptional{As will be seen later, the size of the \ac{LMI} for determining $P$ ..... only grows linearly.}
\end{remark}

\begin{remark}[Interpretability]
	Because each summand \ifpaper{\linebreak[-999]}$\|\Delta A_{\dots}\|_P$ in \cref{theorem:timing-stability:norm-bound} only refers to the timing of at most one sensor and one actuator, its maximum loosely corresponds to the amount of instability caused by the timing of one sensor, actuator or sensor-actuator-pair. This gives important hints on the timing sensitivity, which can be used to improve the design of the real-time system, \eg to give priority to sensors with high sensitivity.
\end{remark}

The following sections present the low-level details of every analysis step.  \Cref{sec:experiments} then shows experimental results.

\section{Discretization} \label{sec:robust-stability:discretization}

\subsection{Definition\xspace \ifreportX{and Discretization\xspace} of a \ac{LIS}} \label{sec:robust-stability:splitting:lis}
A simple definition of a linear impulsive system is
\begin{subequations}
\begin{align}
\dot x(t) &= A\idxContinuous x(t), \quad t \neq \tau_i, \quad t > \tau_0,\\
x(t) &= E_{i} x(t^-),\quad t=\tau_i, \quad i \in \N,\\
x(\tau_0) &= x_0,
\qquad\tau_0 < \tau_1 < \tau_2 < \dots.
\end{align}
\end{subequations}
$A\idxContinuous$ models continuous dynamics, which are interrupted by discrete events $E_i$ at $t=\tau_i$.
For ease of notation, this definition is chosen such that the resulting trajectory is right-continuous, \ie, $x(t^+)=x(t)$.
\ifpaper{\vspace{-1em}}
\paragraph*{Extension to Concurrent Events}
This definition cannot handle concurrent events $\tau_{i}=\tau_{i+1}$, which is a problem for the basic case of perfect timing: In this case, all measurements and actuator updates occur at the same time $t=kT$. \ifpaper{To solve this problem and allow the excluded case $\tau_i=\tau_{i+1}$, we directly define the trajectory as }\ifreport{Therefore, the definition must be extended such that $\tau_{i+1}=\tau_i$ is permitted and leads to the same result as the right-side limit $\tau_{i+1} \to \tau_{i}^+$.

\begin{definition}[\ac{LIS} with Concurrent Events]
	
	A more appropriate generalized definition is the following algorithm, which can be interpreted as a hybrid automaton:
	\begin{enumerate}
		\item Start at $i=0$, $t=\tau_0$, $x(\tau_0)=x_0$.
		\item Compute $x(t)$ for $\tau_i < t \leq \tau_{i+1}$ as solution of $\dot x(t) = A\idxContinuous x(t)$ with known initial value $x(\tau_i)$. (For concurrent events, \ie $\tau_{i+1}=\tau_{i}$, this step has no effect.)
		
		If $\tau_{i+1}$ does not exist because there is only a finite number of events, use the unbounded time range $\tau_i < t < \infty$ instead of $\tau_i < t \le \tau_{i+1}$.
		\label{item:timing-robustness:impulsive-system-cont-step}
		\item Set $x(\tau_{i+1}) \coloneq E_{i+1} x(\tau_{i+1})$ and then set $i \coloneq i+1$. Go to \ref{item:timing-robustness:impulsive-system-cont-step}. (\enquote{Set} refers to overwriting the previous value, analogous to updating a variable in usual (imperative) programming languages.)
	\end{enumerate}
\end{definition}
\todoOptional{später für Diss: Hybriden Automat zeichnen: Init $i=0,t=t_0,x=x_0$, Flow $\dot t=1$, $\dot x = A\idxContinuous x$, Invariant $t \le t_{i+1}$, Transition: Guard $t=t_{i+1}$, Reset $i'=i+1$, $x'=E_i x$}

\paragraph*{Trajectory}
The above algorithm yields an explicit formula for the trajectory of the linear impulsive system:}
\allowdisplaybreaks
\begin{align}
x(t) &\ifreportX{=}\ifpaper{\coloneq} \ee^{A\idxContinuous(t-\tau_{N})} E_N \ee^{A\idxContinuous(\tau_N-\tau_{N-1})} \ifpaper{ \nonumber \\ & \hphantom{==}} E_{N-1}  \ee^{A\idxContinuous(\tau_{N-1}-\tau_{N-2})} \dots E_1 \ee^{A\idxContinuous(\tau_1-\tau_0)} x_0 \\
&= \ifpaper{\textstyle} \ee^{A\idxContinuous(t-\tau_{N})} \left(\reverseProduct_{i=1}^{N} E_i \ee^{A\idxContinuous  (\tau_i - \tau_{i-1})} \right) x(\tau_0) \label{eq:timing-stability:lis-discretization}
\end{align}
with $N$ such that  $\tau_{N} \le t < \tau_{N+1}$ and $\smallReverseProduct$ as defined in \eqref{eq:timing-stability:reverse-product-def}.
\ifpaper{For background, see \specificReferenceToReport{sec:robust-stability:splitting:lis}.}

\todoOptional{Transforming the notation to $\delta_i \coloneq \tau_{i} - \tau_{i-1}$, $D_{i} \coloneq E_i - I$ would save some space later.}

\subsection{Model of Closed Loop as Linear Impulsive System} \label{sec:timing-stability:splitting:lis-control-model}
The closed loop defined in \cref{sec:problem} can be rewritten in the framework of linear impulsive systems, similar to the derivations in \cite{Gaukler2018} and \cite{Rheinfels2019}.
\ifreport{\paragraph*{State}
As noted before, the state is defined as
\begin{equation}
x(t) \coloneq \mat{x\idxPlant(t) \\ x\idxDiscrete(t) \\ y\idxDiscrete(t) \\ u(t)} \in \R^{n}, \quad n=n\idxPlant + n\idxDiscrete + p + m.
\end{equation}}
In the following, all block matrices are separated along the dimensions $n\idxPlant$, $n\idxDiscrete$, $p$, $m$ of the four state components.

\paragraph*{Continuous Dynamics}
The plant dynamics are continuous and all other variables are constant between the discrete events:
\begin{align}
&A\idxContinuous=\mat{A\idxPlant & 0 & 0 & B\idxPlant \\ 0 & 0 & 0 & 0 \\ 0 & 0 & 0 & 0 \\ 0 & 0 & 0 & 0 } \Rightarrow
\ee^{A\idxContinuous \delta}=\mat{\ee^{A\idxPlant\delta} & 0 & 0 & \tilde B(\delta) \\ 0 & I & 0 & 0 \\ 0 & 0 & I & 0 \\ 0 & 0 & 0 & I } \nonumber \\
&\forall \delta \in \R, \text{ with }\tilde B(\delta) \coloneq \ifpaper{\textstyle}\int_0^\delta \ee^{A\idxPlant \xi} \dDiff \xi \,B\idxPlant.
 \label{eq:timing-stability:splitting:lis-control-model:exp-a-tau}
\end{align}

\paragraph*{Discrete Events}
The $k$-th control period is defined as the time range $(k-1/2)T < t \le (k + 1/2)T$. Within this period, all sensors and actuators are updated near $t=kT$:
\ifpaper{\allowdisplaybreaks}
\begin{align}
	E_{\subsActuate,i} &= I + \mat{0 & 0 & 0 & 0 \\ 0 & 0 & 0 & 0 \\ 0 & 0 & 0 & 0 \\ 0 & S_i C\idxDiscrete & 0 &  - S_i}, \quad t_{\subsActuate,i,k}=kT + \Delta t_{\subsActuate,i,k}, \ifpaper{\nonumber\\[-1.3em]} \\
	E_{\subsMeasure,i} &= I + \mat{0 & 0 & 0 & 0 \\ 0 & 0 & 0 & 0 \\ S_iC\idxPlant & 0 & -S_i & 0 \\ 0 & 0 & 0 &  0},\quad t_{\subsMeasure,i,k}=kT+\Delta t_{\subsMeasure,i,k}. \ifpaper{\nonumber\\[-1.3em] }
\end{align}
$S_i \coloneq \ifreport{e_i e_i^\transp =} \mathrm{diag}(\underbrace{0,\dots,0}_{i-1 \text{ times}},1,0,\dots,0)$ are selector matrices of appropriate dimension. The index \enquote{$k$} of the event times will later be omitted for better readability.

Just before the end of the control period, at $t=(k+1/2)T$, the new controller state is computed instantaneously from the recent measurements:
\begin{align}
	E_{\subsCompute} = \mat{I & 0 & 0 & 0 \\ 0 & A\idxDiscrete & B\idxDiscrete & 0 \\ 0 & 0 & I & 0 \\ 0 & 0 & 0 & I},\quad t_{\subsCompute,k}=(k+1/2)T
\end{align}
Note that the actual timing of the controller computation may deviate from this assumption by a bounded amount because updating the controller state has no physical impact. This can be proven by \eqref{eq:timing-stability:lis-discretization} and
\begin{align}
		E_{\subsCompute} \ee^{A\idxContinuous \delta_1} \ee^{A\idxContinuous \delta_2} &= \ee^{A\idxContinuous \delta_1} E_{\subsCompute} \ee^{A\idxContinuous \delta_2} \ifpaper{\nonumber\\&\hspace{-1em}}= 	\ee^{A\idxContinuous \delta_1} \ee^{A\idxContinuous \delta_2} E_{\subsCompute} \quad\forall \delta_{1,2}  \geq 0.
\end{align}
Therefore, the only timing requirements on the controller are its data dependencies: Computation may start as soon all measurements are available and may take until the first actuator is updated.
\ifpaper{\vspace{-.7em}}
\paragraph*{Order of events} 
With $\tau_0 \coloneq kT-T/2$, the set of events $(E_i, \tau_i)$ in the $k$-th control period is
\begin{align}
EV_k \coloneq &\big\lbrace (E_i, \tau_i) | i =1,\dots,N_e\big\rbrace 
\label{eq:timing-stability:order-of-events-1}
\\ = &\left\lbrace (E_{\subsActuate,i}, t_{\subsActuate,i,k}) | i =0,\dots,m-1\right\rbrace \cup \nonumber \\
&\left\lbrace (E_{\subsMeasure,i}, t_{\subsMeasure,i,k}) | i =0,\dots,p-1\right\rbrace \cup \nonumber \\
&\left\lbrace (E_{\subsCompute}, t_{\subsCompute,k}) \right\rbrace \text{ with } \tau_i :\leq \tau_{i+1}, \\
|EV_k| &\coloneq N_e \coloneq m+p+1,
\end{align}
which means that events in each period are numbered as $i=1,\dots,N_e$ according to their temporal order and that all events occur exactly once.
While the order of events with identical time $\tau_i$ is ambiguous, this is not a problem since\ifpaper{, as detailed in \specificReferenceToReport{theorem:timing-stability:reorder},}\ifreport{\xspace the following theorem guarantees that\xspace} all possible orders lead to the same trajectory\ifreport{, thus, an arbitrary order can been chosen without loss of generality}.
\ifreport{
	\begin{theorem}
		The order of actuation and/or measurement events occuring at the same time $\tau_i=\tau_{i+1}$ does not change the system dynamics. \label{theorem:timing-stability:reorder}
	\end{theorem}
	\begin{proof}
		Consider the trajectory \eqref{eq:timing-stability:lis-discretization} of the linear impulsive system. If the $i$-th and $(i+1)$-th event occur at the same time $\tau_i=\tau_{i+1}$, this yields a trajectory 
		$x(t)=\cdots E_{i+1} E_{i} \cdots$. Reversing the order of these events changes the trajectory to $x(t)=\cdots E_{i+1} E_{i} \cdots$. 
		As will be shown later in \eqref{eq:timing-stability:splitting:commutative},  $E_{i+1}E_{i}=E_{i}E_{i+1}$ holds for all measurement and actuation event matrices $E_{i}$, $E_{i+1}$, so the trajectory remains unchanged.
		
	\end{proof}
}

\subsection{CGES $\Leftrightarrow$ DGES} \label{sec:timing-stability:discretization-proof}

In this section, the equivalence of \ac{DGES} and \ac{CGES} will be shown using the fact that the overshoot between two discrete samples is bounded.
\todoOptional{für Report: Skizze mit Abtastwerten und exponentieller Schranke $\bar C \ee^{\bar \lambda \delta}$ dazwischen.}

\begin{theorem} \label{theorem:timing-stability:growth-bound}
	The growth rate of the closed control loop during one control period is bounded:
	
	There exist constants $\bar C \geq 1, \bar \lambda \in \R$ such that $\forall k\geq 0$,
	\begin{equation}
	|x(t_k + \delta)| \le \bar C \ee^{\bar \lambda \delta} |x(t_k)|\quad  \forall \delta \in [0, T), \forall x(t_k) \in \R^n. 	\label{eq:timing-stability:growth-bound}
	\end{equation}
	
	Note that this is not a stability result: Any discrete-time control effectively runs in open loop between the sampling instants, so $\bar \lambda> 0$ if the uncontrolled plant is unstable\ifreport{\xspace(To see that this must be true, consider the case when the plant state is nonzero, \ie, $x\idxPlant(t_k) \neq 0$, and all other entries of $x(t_k)$ are zero)}.
\end{theorem}
\begin{proof}
	Assume $0 < \delta < T$ (the case $\delta=0$ is trivially true). The event matrices from \cref{sec:timing-stability:splitting:lis-control-model} are bounded:
	\begin{equation}
	C_{\mathrm{ev}} \coloneq \max_{M \in \lbrace E_{\subsCompute}, E_{\subsActuate,1}, \dots, E_{\subsActuate,m}, E_{\subsMeasure,1}, \dots, E_{\subsMeasure,p} \rbrace} \|M\|_{\sigma} < \infty
	\end{equation}
	exists because they are constant and finite.
	
	Consider \eqref{eq:timing-stability:lis-discretization} with $N \in \{0, \dots, m+p\}$ as the number of events in $(t_k, t_k + \delta]$. Note that by \eqref{eq:timing-stability:order-of-events-1}, the events are numbered such that the first event after $t= \tau_0 \coloneq t_k$ has the number $i=1$. By \ifreport{\cref{def:norm,theorem:timing-stability:norm-of-exp}}\ifpaper{the properties of the spectral norm},
	\begin{align}
	|x(t_k \ifpaper{&\!} + \ifpaper{\!}\delta)| =\ifreportX{&} \ifpaper{\textstyle }\left|\ee^{A\idxContinuous(t_k + \delta-\tau_{N})} \big(\reverseProduct_{i=1}^{N} \ifpaper{\!} E_i \ee^{A\idxContinuous  (\tau_i - \tau_{i-1})}\big)  x(t_k)\right|  \nonumber\\
	\leq& \ifpaper{\textstyle} \ee^{\|A\idxContinuous \|_{\sigma}(t_k + \delta -\tau_{N})} \reverseProduct_{i=1}^{N} \|E_i\|_{\sigma}\ee^{\|A\idxContinuous \|_{\sigma}(\tau_i-\tau_{i-1})} |x(t_k)| \nonumber\\
	\leq & \ee^{\|A\idxContinuous \|_{\sigma}t_k + \delta - \tau_0} C_{\mathrm{ev}}^N |x(t_k)| \nonumber\\
	\leq& \underbrace{\ee^{\|A\idxContinuous \|_{\sigma}\delta} \strut}_{\ee^{\bar \lambda \delta}} \underbrace{\strut C_{\mathrm{ev}}^{m+p}}_{\bar C}  |x(t_k)|.\ifpaper{\qed}
	\end{align}
\end{proof}
\ifpaper{\vspace{-1.5em}}
\paragraph*{Proof of \cref{theorem:timing-stability:cges-equiv-dges} (\acs{CGES} $\Leftrightarrow$ \acs{DGES}):}
The proof using \cref{theorem:timing-stability:growth-bound} is similar to \cite[Prop.\xspace2]{khatib:16:jnahs}.
\ifpaper{Details are given in \specificReferenceToReport{sec:timing-stability:discretization-proof}.}
\ifreport{
	\xspace A generalized version of this argument is given in \cite{Nesic1999}.
	\paragraph*{\enquote{$\Rightarrow$}:} Assume \ac{CGES}$(\lambda, D)$ and let $\rho = \ee^{\lambda T}$ and $C=D$. Then, the system is \ac{DGES}$(\rho, C)$:
	\begin{equation}
	|x_k|=|x(t_k)| \stackrel{\text{\ac{CGES}}}{\le} D |x(t_0)| \ee^{\lambda kT} = C|x_0| \rho^k.
	\end{equation}
	
	\paragraph*{\enquote{$\Leftarrow$}:} Assume \ac{DGES}$(\rho, C)$, which implies $0<\rho<1$. Let $\lambda = \log(\rho)/T$, so $\lambda < 0$ and $\rho = \ee^{\lambda T}$. Assume $t \geq t_0$, since both \ac{CGES} and \ac{DGES} only refer to this time range.
	Define $k(t):=\left\lfloor (t-t_0)/T\right\rfloor$ as the integer $k$ for which $t_{k(t)} \le t < t_{k(t)+1}$. This implies $k(t) \le (t - t_0)/T$ and therefore $\rho^{k(t)} \le \ee^{\lambda(t-t_0)}$. Because \cref{theorem:timing-stability:growth-bound} bounds the ratio between $x(t)$ and the previous discrete-time sample $x(t_{k(t)})$, the system is \ac{CGES}$(\lambda, D)$:
	\begin{multline}
	|x(t)| \stackrel{\text{\eqref{eq:timing-stability:growth-bound}}}{\leq} |x(t_{k(t)})| \bar C \ee^{\bar \lambda T} \stackrel{\text{\ac{DGES}}}{\leq}  C |x(t_0)| \underbrace{\rho^{k(t)}}_{\le \ee^{\lambda(t-t_0)}} \bar C \ee^{\bar \lambda T} 
	\ifpaper{\\}
	\leq  	\underbrace{C \bar C \ee^{\bar \lambda T}}_{D} \ee^{\lambda (t-t_0)} |x(t_0)|.
	\end{multline}
}
\section{Decomposition}\label{sec:robust-stability:splitting}
In the following, we derive \cref{theorem:timing-stability:splitting}, a key result of our approach: The transition matrix $A_k$ can be split into summands that depend on at most two timing variables.

\ifreport{
	\subsection{Properties of Measurement and Actuation Event Matrices}\label{sec:robust-stability:splitting:matrix-properties}
	In this section, properties of the combinations of event matrices for actuation and measurement will be stated, which will later lead to the proof of \cref{theorem:timing-stability:splitting}. These properties follow directly from block matrix multiplication. For each of the properties, a loose interpretation will be given, which is not to be taken as a formal statement on its own.
	
	\paragraph*{Notation} In the following, $\forall i$ is shorthand for $\forall i \in \left\{1,\dots,m\right\}$ if it refers to $E_{\subsActuate,i}$, and $\forall i \in \left\{1,\dots,p\right\}$ for $E_{\subsMeasure,i}$. The same holds for $\forall j$. Similarly, $\forall \delta$ is shorthand for $\forall \delta \in \R$. The notation $E_{a,\dots } = \dots \forall a \in \left\lbrace \text{\enquote{\ensuremath{\subsActuate}}}, \text{\enquote{\ensuremath{\subsMeasure}}}\right\rbrace$ means that an equation is valid for both $E_{\subsActuate,\dots}$ and $E_{\subsMeasure,\dots}$.
	
	\paragraph*{Properties of a Single Event}
	\begin{lemma}
		Actuation is unaffected by prior delays, as
		\begin{align}
		(E_{\subsActuate,i}-I) \ee^{A\idxContinuous  \delta} &= E_{\subsActuate,i}-I \quad \forall i,\delta  \label{eq:timing-stability:splitting:actuate-time-invariant},
		\end{align}
		whereas measurement is unaffected by \emph{subsequent} delays:
		\begin{align}
		\ee^{A\idxContinuous  \delta} (E_{\subsMeasure,i}-I) &= (E_{\subsMeasure,i}-I) \quad \forall i,\delta. \label{eq:timing-stability:splitting:measurement-reverse-time-invariant}
		\end{align}
	\end{lemma}
	
	\begin{align}
		\intertext{However, measurement \emph{is} affected by prior delays, as}
		(E_{\subsMeasure,i}-I) \ee^{A\idxContinuous  \delta} &= \ifpaper{\small} \mat{0 & 0 & 0 & 0\\ 0 & 0 & 0 & 0 \\ S_i C_p \ee^{A\idxPlant\delta} &  0 & -S_i & S_i C_p \int_0^\delta \ee^{A\idxPlant \xi} d\xi B\idxPlant\\ 0&0&0& 0} \ifpaper{\nonumber \\  & \hphantom{=}}\quad \forall i,\delta
	\end{align}
	\paragraph*{Properties of Two Subsequent Events}
	
	\begin{lemma}[Zero products]
		Products of the form $(E_{\dots} - I) \ee^{A\idxContinuous \delta} (E_{\dots} - I)$ are zero, as long as the events are distinct and the combination is not \enquote{actuate, then measure}:
		\begin{multline}
		(E_{a,i}-I) \ee^{A\idxContinuous  \delta} (E_{b,j} - I)=0 \quad 
		\forall (a,i) \neq (b,j), ~ \forall \delta, \\
		\forall (a,b) \in \left\lbrace \text{\enquote{\ensuremath{\subsActuate}}},
		\text{\enquote{\ensuremath{\subsMeasure}}}\right\rbrace^2 \setminus \left\lbrace (\text{\enquote{\ensuremath{\subsMeasure}}}, \text{\enquote{\ensuremath{\subsActuate}}}) \right\rbrace	 
		\label{eq:timing-stability:splitting:nullproduct}
		\end{multline}
		Additionally, for $\delta=0$, \ie, no delay between the events, this product is always zero:
		\begin{align}
		(E_{a,i}-I) (E_{b,j} - I)=0 \quad &\forall (a,i) \neq (b,j),
		\quad  \ifpaper{\nonumber\\&}
		\forall a,b \in \left\lbrace \text{\enquote{\ensuremath{\subsActuate}}}, \text{\enquote{\ensuremath{\subsMeasure}}}\right\rbrace
		\label{eq:timing-stability:splitting:nullproduct-tau=0}
		\end{align}
	\end{lemma}
	\begin{proof} The lemma directly follows from block matrix computations for each case.
		Actuation of $u_j$ does not affect the subsequent actuation of $u_{i\neq j}$:
		\begin{align}
		(E_{\subsActuate,i}-I) \ee^{A\idxContinuous  \delta} (E_{\subsActuate,j}-I) &= \mat{0 & 0 & 0 & 0\\ 0 &0  & 0 & 0 \\ 0 & 0 & 0 & 0 \\ 0&-S_i S_j C\idxDiscrete&0& S_i S_j} = 0 \quad \forall i\neq j,
		\forall \delta \\
		\intertext{Measurement does not affect subsequent actuation:}
		(E_{\subsActuate,i}-I) \ee^{A\idxContinuous  \delta} (E_{\subsMeasure,j}-I) &= 0 \quad \forall i,j,\delta\\
		\intertext{Measurement of $y_i$ does not affect subsequent measurement of $y_{j\neq i}$:}
		(E_{\subsMeasure,i}-I) \ee^{A\idxContinuous  \delta} (E_{\subsMeasure,j}-I) &=  \mat{0 & 0 & 0 & 0 \\ 0 & 0 & 0 & 0 \\ -S_i S_j C_p & 0 & S_i S_j & 0\\ 0 & 0 & 0 & 0} = 0 \quad \forall i \neq j, \forall \delta.
		\intertext{However, actuation \emph{does} affect subsequent measurements, \ie}
		(E_{\subsMeasure,i}-I) \ee^{A\idxContinuous  \delta} (E_{\subsActuate,j}-I) &= 	\ifpaper{\small} \mat{0 & 0 & 0 & 0 \\ 0 & 0 & 0 & 0 \\ 0 & 0 & 0 & 0 \\ 0 & S_i C_p \int_0^\delta \ee^{A\idxPlant \xi} d\xi B\idxPlant S_j C\idxDiscrete & 0 & -S_i C_p \int_0^\delta \ee^{A\idxPlant \xi} d\xi B\idxPlant S_j}  \quad \forall i,j,\delta
		\intertext{can be nonzero, except if the measurement happens immediately after actuation:}
		(E_{\subsMeasure,i}-I) (E_{\subsActuate,j}-I) &= 0 \quad \forall i,j. \label{eq:timing-stability:splitting:measure-immediately-after-actuate}
		\end{align}
	\end{proof}
	\begin{lemma}[Commutativity]
		All measurement and actuation event matrices commute:
		\begin{equation}
		\forall i,j,\quad \forall a,b \in \left\lbrace \text{\enquote{\ensuremath{\subsActuate}}}, \text{\enquote{\ensuremath{\subsMeasure}}}\right\rbrace: E_{a,i}E_{b,j} = E_{b,j}E_{a,i}. \label{eq:timing-stability:splitting:commutative}
		\end{equation}
	\end{lemma}
	\begin{proof}
		For $(a,i) = (b,j)$, the statement is trivially true. Now consider $(a,i) \neq (b,i)$:
		\begin{align}
		\forall a,b \in &\left\lbrace \text{\enquote{\ensuremath{\subsActuate}}}, \text{\enquote{\ensuremath{\subsMeasure}}}\right\rbrace,
		\quad \forall (a,i) \neq (b,j): \nonumber \\
		E_{a,i}E_{b,j} &= (E_{a,i}-I + I) (E_{b,j} - I + I) \nonumber \\
		&= I + (E_{a,i}-I) + (E_{b,j} - I) + \underbrace{(E_{a,i}-I)(E_{b,j} - I)}_{0 \text{ due to \eqref{eq:timing-stability:splitting:nullproduct-tau=0}}} \nonumber \\
		&= I + (E_{b,j} - I) + (E_{a,i}-I) + 0 \nonumber \\
		&= E_{b,j}E_{a,i}. 
		\end{align}
	\end{proof}

	\paragraph*{Extension to Three and More Events}
	The result \eqref{eq:timing-stability:splitting:nullproduct} leads to a property of the longer chain $(E_{\dots} - I) \ee^{A\idxContinuous \delta} (E_{\dots} - I) \ee^{A\idxContinuous \delta} (E_{\dots} - I)$, again assuming distinct events.
	\begin{lemma}[Long products are zero]
		\begin{align}
		\ifpaper{&}(E_{a,i}-I) \ee^{A\idxContinuous  \delta_1} (E_{b,j} - I)\ee^{A\idxContinuous  \delta_2} (E_{c,k} - I)=0 \quad \ifpaper{\nonumber \\}&\forall (a,b,c) \in \left\lbrace \text{\enquote{\ensuremath{\subsActuate}}}, \text{\enquote{\ensuremath{\subsMeasure}}}\right\rbrace^3,
		\nonumber \\
		&\forall (a,i) \neq (b,j) \neq (c,k), (a,i) \neq (c,k),
		\nonumber \\
		&\forall \delta_1, \delta_2. \label{eq:timing-stability:splitting:nullproduct3}
		\end{align}
		The result implies that any such product of length three and above is zero.
	\end{lemma}
	\begin{proof}
		This is because there are $2^3$ possibilities for $(a,b,c)$, and for each the chain contains at least one product that is zero due to \eqref{eq:timing-stability:splitting:nullproduct}:
		\begin{enumerate}
			\item $\underbrace{(E_{\subsActuate,i}-I) \ee^{A\idxContinuous  \delta_1} (E_{\subsActuate,j} - I)}_{0}\ee^{A\idxContinuous  \delta_2} (E_{\subsActuate,k} - I)=0$
			\item $\underbrace{(E_{\subsActuate,i}-I) \ee^{A\idxContinuous  \delta_1} (E_{\subsActuate,j} - I)}_{0}\ee^{A\idxContinuous  \delta_2} (E_{\subsMeasure,k} - I)=0$
			\item $\underbrace{(E_{\subsActuate,i}-I) \ee^{A\idxContinuous  \delta_1} (E_{\subsMeasure,j} - I)}_{0}\ee^{A\idxContinuous  \delta_2} (E_{\subsActuate,k} - I)=0$
			\item $\underbrace{(E_{\subsActuate,i}-I) \ee^{A\idxContinuous  \delta_1} (E_{\subsMeasure,j} - I)}_{0}\ee^{A\idxContinuous  \delta_2} (E_{\subsMeasure,k} - I)=0$
			\item $(E_{\subsMeasure,i}-I) \ee^{A\idxContinuous  \delta_1} \underbrace{(E_{\subsActuate,j} - I)\ee^{A\idxContinuous  \delta_2} (E_{\subsActuate,k} - I)}_{0}=0$
			\item $(E_{\subsMeasure,i}-I) \ee^{A\idxContinuous  \delta_1} \underbrace{(E_{\subsActuate,j} - I)\ee^{A\idxContinuous  \delta_2} (E_{\subsMeasure,k} - I)}_{0}=0$
			\item $\underbrace{(E_{\subsMeasure,i}-I) \ee^{A\idxContinuous  \delta_1} (E_{\subsMeasure,j} - I)}_{0}\ee^{A\idxContinuous  \delta_2} (E_{\subsActuate,k} - I)=0$
			\item $\underbrace{(E_{\subsMeasure,i}-I) \ee^{A\idxContinuous  \delta_1} (E_{\subsMeasure,j} - I)}_{0}\ee^{A\idxContinuous  \delta_2} (E_{\subsMeasure,k} - I)=0$	
		\end{enumerate}
	\end{proof}
}
\ifreport{
	\subsection{General Expansion of Binomial Products}
	For subsequent proofs, we require a generic way to expand a product of binomials, such as
	\begin{align}
	(A_3 + B_3)(A_2 + B_2)(A_1 + B_1) = &A_3 A_2 A_1 + A_3 A_2 B_1 + A_3 B_2 A_1 + A_3 B_2 B_1 \nonumber \\  &+B_3 A_2 A_1 + B_3 A_2 B_1 + B_3 B_2 A_1 + B_3 B_2 B_1.
	\end{align}
	If \enquote{$A$} and \enquote{$B$} are interpreted as binary digits, where \enquote{$A$} is $0$ and \enquote{$B$} is $1$, the sequence of summands is generated by counting in binary: 000 is $A_3A_2A_1$, 001 is $A_3A_2B_1$, 010 is $A_3 B_2 A_1$, ..., up to 111. To generalize the notation, the binary numbers are interpreted as a vector $\vecSmallT{d_1 & d_2 & d_3}$ of binary digits $d_i \in \{0,1\}$:
	\begin{lemma}
		\label{lemma:robust-stability:splitting:expansion}
		Let $A_1,\dots,A_N,B_1,\dots,B_N \in \R^{n \times n}$ and $N \in \N$. Then,
		\begin{align}
		\reverseProduct_{i=1}^{N}  (A_i + B_i) &= \sum_{\vecSubscriptT{ d_1 & d_2 & \dots & d_N } \in \left\{0,1\right\}^{N}} ~\reverseProduct_{i=1}^{N} 
		\begin{cases}
		A_i, & d_i=0,\\
		B_i, & d_i=1.
		\end{cases}
		\end{align}
	\end{lemma}
	Note that the lemma is given for the reversed product $\smallReverseProduct$, but also valid for the normal product $\Pi$.
	\begin{proof}
		For a proof by induction, we start with the observation that the lemma trivially holds for $N=1$. As induction step, assume the lemma holds for a fixed $N \geq 1$ to show that it holds for $N+1$:
		\begin{align}
		\reverseProduct_{i=1}^{N+1}  (A_i + B_i)  =& (A_{N+1} + B_{N+1}) \reverseProduct_{i=1}^{N}  (A_i + B_i) \\
		\stackrel{\mathclap{\text{assume claim for $N$}}}{=}& \hspace{3em} A_{N+1}  \sum_{\vecSubscriptT{ d_1 & d_2 & \dots & d_N } \in \left\{0,1\right\}^{N}} ~\reverseProduct_{i=1}^{N} 
		\begin{lrdcases}
		A_i, & d_i=0\\
		B_i, & d_i=1
		\end{lrdcases} \\
		&\hspace{3em}+ B_{N+1}  \sum_{\vecSubscriptT{ d_1 & d_2 & \dots & d_N } \in \left\{0,1\right\}^{N}} ~\reverseProduct_{i=1}^{N} 
		\begin{lrdcases}
		A_i, & d_i=0\\
		B_i, & d_i=1
		\end{lrdcases} \\
		=& \sum_{\vecSubscriptT{ d_1 & d_2 & \dots & d_{N+1} } \in \left\{0,1\right\}^{N} \times \{0\},} ~\reverseProduct_{i=1}^{N+1} 
		\begin{lrdcases}
		A_i, & d_i=0\\
		B_i, & d_i=1
		\end{lrdcases} \\
		&+
		\sum_{\vecSubscriptT{ d_1 & d_2 & \dots & d_{N+1} } \in \left\{0,1\right\}^{N} \times \{1\}} ~\reverseProduct_{i=1}^{N+1} 
		\begin{lrdcases}
		A_i, & d_i=0\\
		B_i, & d_i=1
		\end{lrdcases} 
		\\
		=&		\sum_{\vecSubscriptT{ d_1 & d_2 & \dots & d_{N+1} } \in \left\{0,1\right\}^{N+1}} ~\reverseProduct_{i=1}^{N+1} 
		\begin{lrdcases}
		A_i, & d_i=0\\
		B_i, & d_i=1
		\end{lrdcases}
		\end{align}
		By induction, the lemma holds for any $N \geq 1$.
	\end{proof}
}

\ifreportX{\subsection{Proof \ifpaper{Sketch }of \Cref{theorem:timing-stability:splitting}}}
\ifpaper{\paragraph*{Proof \ifpaper{Sketch }of \Cref{theorem:timing-stability:splitting}:}}
 \label{sec:timing-stability:splitting:proof} 
Consider the complete $k$-th control period from $x(t_{k-1})$, \ie,  just after the controller state has been computed, until  $x(t_k)$, \ie just after the next controller computation. As discussed earlier, the period starts with the event counter $i=0$ at $t=\tau_0 \coloneq t_{k-1}=kT-T/2$ and ends after event $i=N_e = m+p+1$ at $t=\tau_{N_e} = t_{k}=kT+T/2$.

\Cref{eq:timing-stability:lis-discretization} leads to $x(t_k) = A_{k-1} x(t_{k-1})$ with 
\newcommand{\dVektor}{\ensuremath{\vecSubscriptT{ d_1 & d_2 & \dots } \in \left\{0,1\right\}^{N_e-1}}}
\begin{align}
	\ifpaper{\!\!\!}A_{k-1} \ifpaper{\!}=& 
	E_{\subsCompute} \ee^{A\idxContinuous(\tau_{N_e}-\tau_{N_e-1})} \underbrace{\ifpaper{\textstyle} \reverseProduct_{i=1}^{N_e-1}  \ifpaper{\!} E_i \ee^{A\idxContinuous  (\tau_i - \tau_{i-1})} }_{\eqcolon X}. \ifpaper{\!\!}
\end{align}
$X$ only contains measurement and actuation events, \ie,\ifreport{\xspace in the following analysis of $X$,}\xspace all matrices $E_i$ are either $E_i=E_{\subsActuate,\dots}$ or $E_i=E_{\subsMeasure,\dots}$. \ifpaper{The remainder of the proof, which can be found in \specificReferenceToReport{sec:timing-stability:splitting:proof}, then consists of expanding the product $X$ and canceling most terms by exploiting the structure of $E_{\{\subsMeasure,\subsActuate\},i}$ and $A\idxContinuous$.
This shows \cref{theorem:timing-stability:splitting} with
{\small
	\begin{align}
	A&(\Delta t = 0)  \nonumber\\[-.5em]
	=& E_{\subsCompute} \ee^{A\idxContinuous T/2} \left(I + \sum_{i=1}^{m} (E_{\subsActuate,i} - I) + \sum_{j=1}^{p} (E_{\subsMeasure,j} - I)\right) \ee^{A\idxContinuous T/2},\nonumber\\[-1.5em] \label{eq:timing-stability:splitting:only-in-paper-result-1}\\
	\Delta &A_{\subsActuate,i} (\Delta t_{\subsActuate,i}) =E_{\subsCompute} \ee^{A\idxContinuous T/2}(\ee^{-A\idxContinuous\Delta t_{\subsActuate,i}}-I) (E_{\subsActuate,i}-I),  \\
	\Delta &A_{\subsMeasure,j} (\Delta t_{\subsMeasure,j})
	=  E_{\subsCompute} (E_{\subsMeasure,j}-I) \ee^{A\idxContinuous T/2}(\ee^{A\idxContinuous \Delta t_{\subsMeasure,j}} - I), \\
	\Delta &A_{\subsActuate\subsMeasure,i,j}(\Delta t_{\subsMeasure,j} - \Delta t_{\subsActuate,i}) \nonumber\\
	=& \begin{cases}
	0, ~~\Delta t_{\subsMeasure,j} - \Delta t_{\subsActuate,i} \leq 0,\\
	E_{\subsCompute} (E_{\subsMeasure,j}-I) \ee^{A\idxContinuous  (\Delta t_{\subsMeasure,j}-\Delta t_{\subsActuate,i})} (E_{\subsActuate,i}-I), & \text{else}.
	\end{cases}\nonumber\\[-1.5em] \label{eq:timing-stability:splitting:only-in-paper-result-4}
	\end{align}
}
All cases of the deviations $\Delta A_{\dots}$ are of the form $M_1 (\ee^{A\idxContinuous \delta(\Delta t)} - I) M_2$, where $M_{1,2}\in\R^{n \times n}$ depend on the event type and $\delta(\Delta t) \in \R$ on the timing such that $\lim_{|\Delta t| \to 0} \delta(\Delta t)=0$. Consequently, we have $\lim_{|\Delta t| \to 0} \Delta A_{\dots} = M_1 (I-I) M_2 = 0$.
}%
\ifreport{
	
	$X$ can be rewritten as
	\begin{align}
		X=&\reverseProduct_{i=1}^{N_e-1}  (I \ee^{A\idxContinuous  (\tau_i - \tau_{i-1})} + (E_i-I) \ee^{A\idxContinuous  (\tau_i - \tau_{i-1})}) \\
		\stackrel{\mathclap{\text{\cref{lemma:robust-stability:splitting:expansion}}}}{=}&\hspace{1em}\sum_{\dVektor} \reverseProduct_{i=1}^{N_e-1} 
		\begin{lrdcases}
		I, & d_i=0\\
		(E_i-I), & d_i=1
		\end{lrdcases} \cdot  \ee^{A\idxContinuous  (\tau_i - \tau_{i-1})}.
	\end{align}	
	This expanded form can be split by $\sum_i d_i$, the amount of how often a factor $(E_i-I)$ appears, to then apply the event matrix properties from \cref{sec:robust-stability:splitting:matrix-properties}:
	\allowdisplaybreaks
	\begin{align}
			X=&\reverseProduct_{i=1}^{N_e-1} I \ee^{A\idxContinuous  (\tau_i - \tau_{i-1})} + \nonumber \\
			&\sum_{j=1}^{N_e-1} \left( \left( \reverseProduct_{i=j+1}^{N_e-1} I  \ee^{A\idxContinuous  (\tau_i - \tau_{i-1})} \right) (E_j-I) \ee^{A\idxContinuous  (\tau_j - \tau_{j-1})} \left( \reverseProduct_{i=1}^{j-1} I \ee^{A\idxContinuous  (\tau_i - \tau_{i-1})}   \right) \right) + \nonumber \\
			&\underbrace{\sum_{\dVektor, \sum_i d_i = 2} \reverseProduct_{i=1}^{N_e-1} \begin{lrdcases}
				I, & d_i=0\\
				(E_i-I), & d_i=1
				\end{lrdcases} \cdot \ee^{A\idxContinuous  (\tau_i - \tau_{i-1})}
			}_{\substack{
				\text{all summands except for the combination } \dots  (E_{\subsMeasure,i}-I) \ee^{A\idxContinuous  \delta} \dots  (E_{\subsActuate,j}-I) \dots \\
				\text{are $=0$ due to \eqref{eq:timing-stability:splitting:nullproduct} and } \ee^{A\idxContinuous \delta_0} I \ee^{A\idxContinuous \delta_1} = \ee^{A\idxContinuous(\delta_0+\delta_1)}
			}}
			 + \nonumber \\
			&\underbrace{
				\sum_{\dVektor, \sum_i d_i \geq 3} \reverseProduct_{i=1}^{N_e-1}   \begin{lrdcases}
				I, & d_i=0\\
				(E_i-I), & d_i=1
				\end{lrdcases} \cdot \ee^{A\idxContinuous  (\tau_i - \tau_{i-1})}
			}_{ =0 \text{ due to \eqref{eq:timing-stability:splitting:nullproduct3} and }\ee^{A\idxContinuous \delta_0} I \ee^{A\idxContinuous \delta_1} = \ee^{A\idxContinuous(\delta_0+\delta_1)}}.\\
		=& \reverseProduct_{i=1}^{N_e-1}  \ee^{A\idxContinuous  (\tau_i - \tau_{i-1})} + \nonumber \\
		&\sum_{j=1}^{N_e-1} \left( \left( \reverseProduct_{i=j+1}^{N_e-1} I \ee^{A\idxContinuous  (\tau_i - \tau_{i-1})}  \right)  (E_j-I)  \ee^{A\idxContinuous  (\tau_j - \tau_{j-1})}\left( \reverseProduct_{i=1}^{j-1} I \ee^{A\idxContinuous  (\tau_i - \tau_{i-1})}  \right)  \right)+ \nonumber \\
		&\sum_{i=1}^{p} \sum_{j=1}^{m} \small \begin{cases}
		0, & t_{\subsActuate,j} \geq t_{\subsMeasure,i}\\
		\ee^{A\idxContinuous  (\tau_{N_e-1}-t_{\subsMeasure,i})} (E_{\subsMeasure,i}-I) \ee^{A\idxContinuous  (t_{\subsMeasure,i}-t_{\subsActuate,j})} (E_{\subsActuate,j}-I) \underbrace{\ee^{A\idxContinuous  (t_{\subsActuate,j}-\tau_0)}}_{\text{omit due to \eqref{eq:timing-stability:splitting:actuate-time-invariant}}}, & \text{else}
		\end{cases} \\
		=& \ee^{A\idxContinuous  (\tau_{N_e-1}-\tau_0)} + \sum_{j=1}^{N_e-1} \ee^{A\idxContinuous(\tau_{N_e-1} - \tau_{j})} (E_j-I) \ee^{A\idxContinuous(\tau_{j} - \tau_0)} + \nonumber \\
			&\sum_{i=1}^{p} \sum_{j=1}^{m} \small \begin{cases}
		0, & t_{\subsActuate,j} \geq t_{\subsMeasure,i}\\
		\ee^{A\idxContinuous  (\tau_{N_e-1}-t_{\subsMeasure,i})} (E_{\subsMeasure,i}-I) \ee^{A\idxContinuous  (t_{\subsMeasure,i}-t_{\subsActuate,j})} (E_{\subsActuate,j}-I), & \text{else}
		\end{cases} \\
		=& \ee^{A\idxContinuous  (\tau_{N_e-1}-\tau_0)} + \nonumber \\
		& \sum_{i=1}^{m} \ee^{A\idxContinuous(\tau_{N_e-1} - t_{\subsActuate,i})} (E_{\subsActuate,i}-I) \underbrace{\ee^{A\idxContinuous(t_{\subsActuate,i} - \tau_0)}}_{\text{omit due to \eqref{eq:timing-stability:splitting:actuate-time-invariant}}} + \nonumber \\
		& \sum_{j=1}^{p} \ee^{A\idxContinuous(\tau_{N_e-1} - t_{\subsMeasure,j})} (E_{\subsMeasure,j}-I) \ee^{A\idxContinuous(t_{\subsMeasure,j} - \tau_0)} + \nonumber \\	
		&\sum_{i=1}^{p} \sum_{j=1}^{m} \small \begin{cases}
		0, & t_{\subsActuate,j} \geq t_{\subsMeasure,i}\\
		\ee^{A\idxContinuous  (\tau_{N_e-1}-t_{\subsMeasure,i})} (E_{\subsMeasure,i}-I) \ee^{A\idxContinuous  (t_{\subsMeasure,i}-t_{\subsActuate,j})} (E_{\subsActuate,j}-I), & \text{else.}
		\end{cases}
		\intertext{According to this splitting of $X$, $A_k=A_{\subsCompute} \ee^{A\idxContinuous(\tau_{N_e}-\tau_{N_e-1})} X$ can be rewritten as}
	A_k &= E_{\subsCompute} \ee^{A\idxContinuous(\tau_{N_e}-\tau_0)} + \nonumber \\
	& \sum_{i=1}^{m}E_{\subsCompute} \ee^{A\idxContinuous(\tau_{N_e} - t_{\subsActuate,i})} (E_{\subsActuate,i}-I)+ \nonumber \\
	& \sum_{j=1}^{p} E_{\subsCompute} \underbrace{\ee^{A\idxContinuous(\tau_{N_e} - t_{\subsMeasure,j})}}_{\text{omit due to \eqref{eq:timing-stability:splitting:measurement-reverse-time-invariant}}} (E_{\subsMeasure,j}-I) \ee^{A\idxContinuous(t_{\subsMeasure,j} - \tau_0)} + \nonumber \\	
	&\sum_{i=1}^{p} \sum_{j=1}^{m} \small \begin{cases}
	0, & t_{\subsActuate,j} \geq t_{\subsMeasure,i}\\
	E_{\subsCompute} \underbrace{\ee^{A\idxContinuous(\tau_{N_e} - t_{\subsMeasure,j})}}_{\text{omit due to \eqref{eq:timing-stability:splitting:measurement-reverse-time-invariant}}} (E_{\subsMeasure,i}-I) \ee^{A\idxContinuous  (t_{\subsMeasure,i}-t_{\subsActuate,j})} (E_{\subsActuate,j}-I), & \text{else}.
	\end{cases}
	\label{eq:timing-stability:splitting:Ak-unsimplified}
	\\
	\intertext{With $t_{\{u,y\},i} = \Delta t_{\{u,y\},i} + kT$, this becomes}
	A_k &= E_{\subsCompute} \ee^{A\idxContinuous T} + \nonumber \\
	& \sum_{i=1}^{m}\underbrace{ E_{\subsCompute} \ee^{A\idxContinuous(T/2 - \Delta t_{\subsActuate,i})} (E_{\subsActuate,i}-I)}_{M_{u,i}(\Delta t_{\subsActuate,i})} + \nonumber \\
	& \sum_{j=1}^{p} \underbrace{ E_{\subsCompute} (E_{\subsMeasure,j}-I) \ee^{A\idxContinuous(T/2 + \Delta t_{\subsMeasure,j})} }_{M_{\subsMeasure,j}(\Delta t_{\subsMeasure,j})} + \nonumber \\	
	&\sum_{j=1}^{p} \sum_{i=1}^{m} \small \underbrace{\begin{cases}
	0, &  \Delta t_{\subsMeasure,j} - \Delta t_{\subsActuate,i} \leq 0\\
	E_{\subsCompute} (E_{\subsMeasure,j}-I) \ee^{A\idxContinuous  (\Delta t_{\subsMeasure,j}-\Delta t_{\subsActuate,i})} (E_{\subsActuate,i}-I), & \text{else}.
	\end{cases}}_{M_{\subsActuate\subsMeasure,i,j}(\Delta t_{\subsMeasure,j} - \Delta t_{\subsActuate,i})}\\
	&= E_{\subsCompute} \ee^{A\idxContinuous T} + \sum_{i=1}^{m} M_{\subsActuate,i}(\Delta t_{\subsActuate,i}) + \sum_{i=1}^{p} M_{\subsMeasure,i}(\Delta t_{\subsMeasure,i}) + \sum_{i=1}^{m} \sum_{j=1}^{p} M_{\subsActuate\subsMeasure,i,j}(\Delta t_{\subsMeasure,j} - \Delta t_{\subsActuate,i})
	\intertext{Setting this equal to the desired result of \cref{theorem:timing-stability:splitting},}
		A_k \stackrel{!}{=}& A(\Delta t = 0) +  \sum_{i=1}^{m} \underbrace{\Delta A_{\subsActuate,i} (\Delta t_{\subsActuate,i})}_{\text{small}} +  \sum_{j=1}^{p} \underbrace{\Delta A_{\subsMeasure,j} (\Delta t_{\subsMeasure,j})}_{\text{small}} + \nonumber \\&
	\sum_{i=1}^{m} \sum_{j=1}^{p}  \underbrace{\Delta A_{\subsActuate\subsMeasure,i,j} (\Delta t_{\subsMeasure,j} - \Delta t_{\subsActuate,i})}_{\text{small}}, \label{eq:timing-stability:splitting-repeated-in-proof}
	\end{align}
	leads to
	\begin{align}
	A(\Delta t = 0) =& \left.A_k\right|_{\Delta t_{\subsMeasure,0,1,\dots,m} = \Delta t_{\subsActuate,0,1,\dots,p}=0}\\
	=& E_{\subsCompute} \ee^{A\idxContinuous T} + E_{\subsCompute} \ee^{A\idxContinuous T/2} \left(\sum_{i=1}^{m} (E_{\subsActuate,i} - I) \right) + \nonumber \\
	&E_{\subsCompute} \left(\sum_{j=1}^{p} (E_{\subsMeasure,j} - I)\right) \ee^{A\idxContinuous T/2} \\
	\stackrel{\text{\eqref{eq:timing-stability:splitting:actuate-time-invariant}, \eqref{eq:timing-stability:splitting:measurement-reverse-time-invariant}}}{=}& E_{\subsCompute} \ee^{A\idxContinuous T/2} \left(I + \sum_{i=1}^{m} (E_{\subsActuate,i} - I) + \sum_{j=1}^{p} (E_{\subsMeasure,j} - I)\right) \ee^{A\idxContinuous T/2},  \label{eq:timing-stability:splitting:only-in-report-result-1}\\
	\Delta A_{\subsActuate,i} (\Delta t_{\subsActuate,i}) =& M_{\subsActuate,i}(\Delta t_{\subsActuate,i}) -  M_{\subsActuate,i}(0) \\
	=& E_{\subsCompute} \ee^{A\idxContinuous T/2}(\ee^{-A\idxContinuous\Delta t_{\subsActuate,i}}-I) (E_{\subsActuate,i}-I)  	\label{eq:timing-stability:splitting:only-in-report-result-2}\\
	\Delta A_{\subsMeasure,j} (\Delta t_{\subsMeasure,j}) =& M_{\subsMeasure,j}(\Delta t_{\subsMeasure,j}) -  M_{\subsMeasure,j}(0) \\
	=&  E_{\subsCompute} (E_{\subsMeasure,j}-I) \ee^{A\idxContinuous T/2}(\ee^{A\idxContinuous \Delta t_{\subsMeasure,j}} - I), 	\label{eq:timing-stability:splitting:only-in-report-result-3}\\
	\Delta A_{\subsActuate\subsMeasure,i,j}(\Delta t_{\subsMeasure,j} - \Delta t_{\subsActuate,i}) =& M_{\subsActuate\subsMeasure,i,j}(\Delta t_{\subsMeasure,j} - \Delta t_{\subsActuate,i}) -  \underbrace{M_{\subsActuate\subsMeasure,i,j}(0,0)}_0 \\
	=& \begin{cases}
	0, & \Delta t_{\subsMeasure,j} - \Delta t_{\subsActuate,i} \leq 0\\
	E_{\subsCompute} (E_{\subsMeasure,j}-I) \ee^{A\idxContinuous  (\Delta t_{\subsMeasure,j}-\Delta t_{\subsActuate,i})} (E_{\subsActuate,i}-I), & \text{else}.
	\end{cases}
	\end{align}
	
	This proves the key equation of \cref{theorem:timing-stability:splitting}. With
	\begin{equation}
		(E_{\subsMeasure,j}-I) (M-I) (E_{\subsActuate,i}-I) = (E_{\subsMeasure,j}-I) M (E_{\subsActuate,i}-I) - \underbrace{(E_{\subsMeasure,j}-I) (E_{\subsActuate,i}-I)}_{=0 \text{ due to \eqref{eq:timing-stability:splitting:nullproduct}}} \quad \forall M \in \R^{n \times n},
	\end{equation}
	$\Delta A_{\subsActuate\subsMeasure,i,j}$ can be rewritten as
	\begin{align}
	\Delta A_{\subsActuate\subsMeasure,i,j}(\Delta t_{\subsMeasure,j} - \Delta t_{\subsActuate,i}) =& \begin{cases}
	0, & \Delta t_{\subsMeasure,j} - \Delta t_{\subsActuate,i} \leq 0\\
	 {E_{\subsCompute} (E_{\subsMeasure,j}-I) (\ee^{A\idxContinuous  (\Delta t_{\subsMeasure,j}-\Delta t_{\subsActuate,i})} - I) (E_{\subsActuate,i}-I),} & \text{else}.
	\end{cases}
	\label{eq:timing-stability:splitting:only-in-report-result-4}
	\end{align}
	Now, all $\Delta A_{\dots}$ are of the form $M_1 (\ee^{A\idxContinuous \delta} - I) M_2$, which simplifies the derivation of bounds. Note that in all cases, $\delta \to 0$ if $|\Delta t| \to 0$. Therefore, $\lim_{|\Delta t| \to 0} \Delta A_{\dots} = M_1 (I-I) M_2 = 0$.
	
	This concludes the proof of \cref{theorem:timing-stability:splitting}.
}%
The results \ifpaper{\eqref{eq:timing-stability:splitting:only-in-paper-result-1}--\eqref{eq:timing-stability:splitting:only-in-paper-result-4}}\ifreportX{\labelcref{eq:timing-stability:splitting:only-in-report-result-1,eq:timing-stability:splitting:only-in-report-result-2,eq:timing-stability:splitting:only-in-report-result-3,eq:timing-stability:splitting:only-in-report-result-4}} are validated by numerical experiments. \ifreport{(For details, see  \texttt{check\_Ak\_delta\_to\_nominal()} in \texttt{src/qronos/lis/test\_lis.py}  in the code linked in \cref{sec:experiments}.)}

\section{$P$-Ellipsoid Norm} \label{sec:timing-stability:norm}
This section presents connections between the Lyapunov candidate function $V_P(x):=x^\transp P x$ and the $P$-ellipsoid matrix norm.
\ifpaper{\vspace{-1.5em}}
\ifreport{
	\begin{theorem} \label{theorem:timing-stability:lyapunov-dt-lti}		
		Iff the \emph{time-invariant} system $x_{k+1} = A x_k$ is exponentially stable (\ie, \ac{DGES}), there always exists a quadratic Lyapunov function $V_P(x):=x^\transp P x \succ 0$ with $V_P(Ax) \prec V_P(x)$. 
		$P \in \R^{n \times n}$ 
		is a solution of
		\begin{equation}
		A^\transp P A - P = -Q, \qquad P,Q \succ 0
		\end{equation}
		with the positive definite parameter $Q$	
		\cite[Proposition 11.10.5]{Bernstein2009}.
	\end{theorem}
}%
	\ifreport{
		
		Note that for all $P\succ 0$,  $\|A\|_P < 1$ is equivalent to $V_P(Ax) \prec V_P(x)$. The $P$-ellipsoid norm can therefore be interpreted as the matrix norm which is equivalent to a quadratic Lyapunov function. The relation of this norm to the (joint) spectral radius is further discussed in \cite[Section 2.3.7]{Jungers2009} and \cite{Blondel2005}.
	}%
%
\paragraph*{Proof of \cref{theorem:timing-stability:p-norm} ($\|\cdot\|_P$ is a submultiplicative norm):}
	Since $P\succ 0$, $\sqrt{V_P(x)} = \sqrt{x^{\transp} P x}$ is a vector norm \cite[Fact 9.7.30]{Bernstein2009}. The $P$-ellipsoid norm $\|\cdot\|_P$ is its equi-induced matrix norm,  therefore submultiplicative\ifreport{\xspace with $\|I\|_P=1$ due to \cref{def:timing-stability:equi-induced-norm,theorem:timing-stability:norm-of-exp}}. \ifreport{\qed}

\ifreport{The vector norm $\sqrt{V_P(x)}$ can also be seen as the euclidean norm after applying a coordinate transformation, as
	\begin{equation}
	\sqrt{V_P(x)}=\sqrt{x^{\transp} P x} = \sqrt{x^{\transp} P^{1/2} (P^{1/2})^{\transp} x} = |\underbrace{(P^{1/2})^{\transp}x}_{z}|,
	\end{equation}
	where $P^{1/2}$ is the Cholesky decomposition of $P$ per \cref{theorem:timing-stability:cholesky}.
	If $V_P$ is a Lyapunov function, the transformed system is contractive, \ie, $|z_{k+1}| = \sqrt{V_P(x_{k+1})} \le \sqrt{V_P(x_{k})} = |z_k|$.}

\begin{theorem} 	\label{theorem:timing-stability:p-norm-via-spectral}
$
		\|A\|_P = \|(P^{1/2})^{\transp} A (P^{1/2})^{-\transp}\|_{\sigma}.
$
\end{theorem}
\begin{proof}
	Rewrite the $P$-ellipsoid norm as 
	\begin{equation}
	\|A\|_P = \max_{x \neq 0} \fracIfReport{|(P^{1/2})^{\transp} A x|}{|(P^{1/2})^{\transp} x|}
	\end{equation}
	 and change variables to $z$ with $x=(P^{1/2})^{-\transp} z$:
	\begin{multline}
	\ifreportX{\hfill}
	\|A\|_P  = \max_{z \neq 0} \fracIfReport{|(P^{1/2})^{\transp} A (P^{1/2})^{-\transp} z|}{|z|} \ifpaper{\\[-.3em]}= \|(P^{1/2})^{\transp} A (P^{1/2})^{-\transp}\|_{\sigma} \ifpaper{.}\ifreportX{\hfill}\ifpaper{\qed}
	\end{multline}
\end{proof}

\begin{theorem}[Extreme Quadratic Lyapunov Function]
	\label{theorem:timing-stability:extreme-lyapunov-lti}
	If a time-invariant system $x_{k+1} = A x_k$ is stable, \ie, $\spectralRadius{ A } < 1$, then there exists a quadratic Lyapunov function $V_P(x)$ that proves a stability factor $\bar \rho$ arbitrarily close to  the spectral radius $\spectralRadius{ A }$:
	\begin{multline}
	\forall A\in \R^{n \times n} \text{ with }\spectralRadius{A} < 1  \quad \forall \bar\rho > \spectralRadius{A} \ifpaper{\\} \quad \exists P \quad  \|A\|_P = \max_{x\neq 0} \sqrt{\fracIfReport{V_P(Ax)}{V_P(x)}} \le \bar \rho.
	\end{multline}
\end{theorem}

\begin{proof}
	\ifpaper{See \specificReferenceToReport{theorem:timing-stability:extreme-lyapunov-lti}. \ifpaper{\qed}}
	\ifreport{
		Assume $\spectralRadius{A} < 1$ and $\bar \rho > \spectralRadius{A}$. Assume that $\bar \rho < 1$, which is without loss of generality because the resulting $P$ is also valid for any $\bar \rho>1$. 
		
		Consider the \enquote{destabilized} system $\tilde x_{k+1} = A \bar \rho^{-1} \tilde x_k$, for which all eigenvalues and therefore the spectral radius are scaled by $\bar \rho^{-1}$. It is still stable, but almost unstable for $\bar\rho \to \spectralRadius{A}^+$.
		\begin{equation}
		\spectralRadius{ A \bar \rho^{-1} } = \bar \rho^{-1} \spectralRadius{ A }  \in (\spectralRadius{A},1)
		\end{equation}
		
		Applying \cref{theorem:timing-stability:lyapunov-dt-lti} to $A \bar \rho^{-1}$ and any $Q\succ 0$ shows that there is a $P$ such that
		\begin{align}
		V_{P}(\bar \rho^{-1}Ax) &\prec V_{ P}(x) \\
		\Rightarrow \bar \rho^{-2} V_{ P} (Ax) &\prec V_{ P}(x) \\
		\Rightarrow \frac{V_{P} (Ax)}{V_{ P}(x)} & \prec \bar \rho^{2} \quad \forall x \neq 0 \\
		\Rightarrow \sqrt{\frac{V_P(Ax)}{V_P(x)}} &\le \bar \rho  \quad \forall x \neq 0\\
		\Rightarrow \|A\|_P &\leq \bar\rho.
		\end{align}
		Note that the resulting $V_P(x)$ is a Lyapunov function for both $A \bar \rho^{-1}$ and $A$.
	}
\end{proof}

\ifreport{
	\begin{theorem}[Spectral Radius Bound via Matrix Norms {\cite[Proposition 2.6]{Jungers2009}}]
		\label{theorem:timing-stability:jsr-from-matrixnorm}
		Any submultiplicative matrix norm $\|\cdot\|_S$ leads to an upper bound of the spectral radius:
		\begin{equation}
		\spectralRadius{A} \leq \|A\|_S \quad  \forall A \in \R^{n \times n} \quad \forall  \|\cdot\|_S.
		\end{equation}
		\todoOptional{Für Diss: Das hier ist die von JSR auf SR vereinfachte Version von jsr-from-matrixnorm; hier wieder entfernen sobald das JSR-Kapitel eingebunden ist. Die Zitation ist hier auch nur 95\% richtig, streng genommen fehlt noch der Zusammenhang JSR=SR.}
	\end{theorem}
	\begin{remark}[Extremal $P$-ellipsoid norm]
		\label{remark:timing-stability:extreme-p-norm}
		In the general case, there is no lower $P$-ellipsoid norm than the one guaranteed by \cref{theorem:timing-stability:extreme-lyapunov-lti}. Especially, it is not generally possible to find a $P$ such that $\|A\|_P=\spectralRadius{A}$ holds exactly.
	\end{remark}
	\begin{proof}
		Due to \cref{theorem:timing-stability:jsr-from-matrixnorm}, $\|A\|_P \geq \spectralRadius{A}$ always holds. The remainder of this proof is to show by example that \enquote{$=$} is not generally possible, \ie,
		\begin{equation}
		\text{for }A=\mat{\rho & 1 \\ 0 & \rho}\text{ with }0 < \rho < 1,\quad \|A\|_P \neq \spectralRadius{A} \quad \forall P\succ 0.
		\end{equation}

		Assume $A$ as given in the previous equation. Here, $\spectralRadius{A}=\rho < 1$. All $P \succ 0$  can be parameterized using \cref{theorem:timing-stability:cholesky} as 
		\begin{equation}
		P= P^{1/2} (P^{1/2})^{\transp} \text{ with } P^{1/2}=\mat{a & 0 \\ b & c}, ~ a>0,~c>0,~b\in\R.
		\end{equation}
		By \cref{theorem:timing-stability:p-norm-via-spectral}, 
		\begin{equation}
		\|A\|_P=\|\underbrace{(P^{1/2})^{\transp} A (P^{1/2})^{-\transp}}_M\|_{\sigma} = \max\{\sigma_1, \sigma_2\},
		\end{equation}
		where $\sigma_i^2$ are the eigenvalues of $M^\transp M$, which are the solutions of
		\begin{align}
		0 &= \det (M^\transp M - \sigma_i^2 I). \\
		\Leftrightarrow \dots \Leftrightarrow
		0 &= {\sigma _{i}}^4-\underbrace{\left(\frac{a^2}{c^2} + 2\rho ^2\right)}_{\eqcolon d > 2 \rho^2 > 0}\,{\sigma _{i}}^2+\rho ^4\\
		\Leftrightarrow \sigma_i^2& = \frac{\overbrace{d}^{>2\rho^2} \pm ({\overbrace{d^2 - 4 \rho^4}^{>0}})^{1/2}}{2}  \quad \Rightarrow \max_{i \in \{1,2\}} \sigma_i^2 > \rho^2\\
		\Rightarrow \|A\|_P &> \spectralRadius{A} \quad \text{ for all possible }P.
		\end{align}
		In the limit $\|A\|_P \to \spectralRadius{A}$, this results in $c\to\infty$ or $a \to 0$, so that $P^{1/2}$ or $(P^{1/2})^{-1}$ become numerically problematic. This motivates that a numerical solution for $P$ should stay away from this limit, but rather keep some distance $\|A\|_P - \spectralRadius{A} > 0$ to ensure numerical robustness.
		
		The source code linked in \cref{sec:experiments} contains symbolic and numeric computations for this example in \texttt{notes/matlab\_counterexample\_for\_existence\_of\_extreme\_P.m}.
	\end{proof}
}

\begin{theorem}[Robust stability from norm bounds]
	\label{theorem:timing-stability:norm-bound:generic}
	Let\ifpaper{\hphantom{\,}} $A_k = \sum_{i=0}^{N} A_{k,i}$ with fixed $N$. Then, the system $x_{k+1} = A_k x_{k}$ is \ac{DGES}$(\bar\rho, C)$ for some $C$ if there are a submultiplicative matrix norm $\|\cdot\|$ and a bound $0 \leq \bar \rho<1$ such that $\sum_i \|A_{k,i}\| \le \bar \rho ~ \forall k$.
\end{theorem}
\begin{proof}
	{
		Assume $\sum_i \|A_{k,i}\| \le \bar \rho  < 1~ \forall k$. The triangle inequality\ifreport{\xspace \eqref{eq:timing-stability:norm-additive}\xspace} leads to $\|A_k\| = \|\sum_{i=0}^{N} A_{k,i}\| \le \sum_{i=0}^{N} \|A_{k,i}\| \le \bar \rho$. Due to \ifreport{\cref{theorem:timing-stability:norm-comparison}}\ifpaper{the equivalence of norms}, there is a finite $C>0$ such that $\|M\|_{\sigma} \le C \|M\|$ for all $M \in \R^{n \times n}$. This leads to
		\begin{align} 
		|x_{k+1}| =& \ifpaper{\textstyle} \left|\left(\reverseProduct_{j=0}^{k} A_{j}\right) x_0 \right| \le   \left\|\reverseProduct_{j=0}^{k} A_{j}\right\|_{\sigma} |x_0| \ifpaper{\nonumber\\} \le \ifpaper{&} \ifpaper{\textstyle} C \left\|\reverseProduct_{j=0}^{k} A_{j}\right\| |x_0| \le C \bar \rho^k |x_0| \quad \forall x_0 \in \R^n,
		\end{align}
		which proves \ac{DGES}$(\bar\rho, C)$.\ifpaper{\qed}
	}
\end{proof}

\section{Norm bounding of summands}
\label{sec:timing-stability:normbound}
\ifreport{
	\subsection{Bound on Timing-Dependent Deviations}
	\label{sec:timing-stability:normbound:expm}
}

\Cref{theorem:timing-stability:norm-bound} provides a stability result based on the $P$-ellipsoid norm of the timing-dependent deviations $\Delta A_{\dots}$. In this section, a bound for this norm is presented using the general form $\Delta A_{\dots}=M_1 (\ee^{A \tau} - I) M_2$ shown in \cref{sec:timing-stability:splitting:proof}. 

\ifpaper{By \specificReferenceToReport{sec:timing-stability:normbound:expm}, a Taylor series of order $r$ yields}
\ifreport{
	\paragraph*{Problem Statement} For small $\delta$, compute a bound on $\max\limits_{\tau \in [-\delta ,\delta]}\|M_1(\ee^{A \tau}-I)M_2\|_P$.
	
	
	\paragraph*{Idea} 
	 A series expansion of the matrix exponential
	\begin{align}
		\ee^{A\tau} -I =& \sum_{i=0}^{\infty} \frac{A^i \tau^i}{i!}  - I = \sum_{i=1}^{r} \frac{A^i \tau^i}{i!} + \underbrace{\sum_{i=r+1}^{\infty} \frac{A^i \tau^i}{i!}}_{E}
	\end{align}
	is expanded up to order $r \geq 0$, and the remainder $E$ is bounded.
	
	\paragraph*{Implementation} Applying this idea leads to
	\begin{align}
			\nonumber
			\|M_1(\ee^{A\tau}-I)M_2\|_P=& \left\|\sum_{i=1}^{r} \frac{M_1 A^i M_2 \tau^i}{i!} + M_1\ifpaper{\!}\sum_{i=r+1}^{\infty} \ifpaper{\!\!}\frac{A^i \tau^i}{i!} M_2 \right\|_P \\
			\leq& \sum_{i=1}^{r} \frac{\|M_1 A^i M_2\|_P |\tau|^i}{i!} \ifpaper{ \nonumber \\&} + \|M_1\|_P \sum_{i=r+1}^{\infty} \frac{ \|A \|^i_P  |\tau|^i}{i!} \|M_2\|_P \ifpaper{ \nonumber \\}\eqcolon \ifpaper{&}\,h(|\tau|).
		\end{align}
	
		As $h(|\tau|)$ is a polynomial of $|\tau|$ with nonnegative coefficients, it is nondecreasing for increasing $|\tau|$. Therefore, its bounds for $|\tau| \in [0,\delta]$ are $h(0)=0$ and $h(\delta)$:
}
\begin{align}
	0 \le& \|M_1(\ee^{A\tau}-I)M_2\|_P \le h(\delta) \quad \forall \tau \in [-\delta, \delta]\ifreportX{.}
\end{align}
	\ifreport{For computation, $h(\delta)$ with $\delta \geq 0$ is rewritten as}
	\ifreportStartMark
	\begin{align}
	\ifpaper{&\text{with~}} 
		h(\delta)
		\ifreportX{
			=& -\|M_1\|_P \|M_2\|_P+ \sum_{i=1}^{r} \delta^i \underbrace{\ifpaper{\textstyle }\frac{\|M_1 A^i M_2\|_P - \|M_1\|_P\|A\|^i_P\|M_2\|_P}{i!}}_{\eqcolon \gamma_i, ~ i \geq 1} \nonumber\\& +\|M_1\|_P \underbrace{\sum_{i=0}^{\infty} \frac{ \|A \|^i_P  \delta^i}{i!}}_{\ee^{\|A\|_P\,\delta}} \|M_2\|_P \\
		}
		=\ifreportX{&}  \|M_1\|_P\|M_2\|_P (\ee^{\|A\|_P\,\delta} - 1) + \ifpaper{\textstyle} \sum_{i=1}^{r} \gamma_i \delta^i \ifreportX{.}
		\ifpaper{,\nonumber\\&\gamma_i \coloneq \left({\|M_1 A^i M_2\|_P - \|M_1\|_P\|A\|^i_P\|M_2\|_P}\right)/(i!).}
	\end{align}
	\ifreportEndMark
	\ifreport{(To include the special case of $\delta=0$, the above derivation uses the definition $0^0\coloneq 1$.)\xspace} 
As $\lim_{\delta \to 0^+} h(\delta) = 0$, this bound preserves the property 
\begin{equation}
	\|\Delta A_{\dots}\|_P \to 0 \text{ for }\Delta t \to 0
\end{equation}
from \cref{theorem:timing-stability:splitting}, and therefore also the feasibility result from \cref{theorem:timing-stability:feasibility}. In the implementation, $r=10$ is used.
\todoOptional{Später: Wieso funktioniert diese Approximation für Beispiel B2 (rein akademisch, hier nicht gezeigt) nicht gescheit?}
\todoOptional{Wer das liest kriegt ein Eis, solange Vorrat reicht :-D}

\ifreport{\subsection{Verified Numerical Implementation}
\label{sec:timing-stability:norm-bound:numeric}}
To ensure a safe overapproximation despite finite numerical precision, interval arithmetic is used to determine all norms and norm bounds.
\xspace
\ifpaper{The interval computation of norms, based on \cite{Rump2010}, is explained in detail in \specificReferenceToReport{sec:timing-stability:norm-bound:numeric}.}%
\ifreport{%
	This leads to an overapproximated, \ie, pessimistic but guaranteed result.
	
	The numerical approximation of $P^{1/2}$ results in an approximate value $K \neq P^{1/2}$ without guarantees on the distance $K-P^{1/2}$ to some \enquote{nearest} valid solution for $P^{1/2}$. Let $\tilde P = KK^\transp$ be the corresponding replacement for $P$. If there is a bound $\bar \rho$ such that $\|A_k(\Delta t_{k}=0)\|_{\tilde P} < \bar \rho < 1$, this approximation is usable to show stability for some timing bounds. Otherwise, the stability analysis has failed.
	
	Computing a guaranteed bound for $\bar \rho$ despite numerical errors is possible using \emph{interval arithmetic} in the computation of $\|A_k(\Delta t_{k}=0)\|_{\tilde P}$ via the spectral norm, as will be explained later. 
	
	To show stability using \cref{theorem:timing-stability:norm-bound}, $\sum \|\Delta A_{\dots}\|_{\tilde P} < 1-\bar \rho$ must be checked. A bound on each summand is computed by evaluating \cref{sec:timing-stability:normbound:expm} in interval arithmetic.
	
	The use of interval arithmetic has the advantage that small uncertainties in the plant model $A\idxPlant,B\idxPlant,C\idxPlant$ and the period $T$ can be explicitly considered. Because the result of the presented approach is a \acl{CQLF}, the same stability result also holds if the uncertainties are time-varying (\cref{theorem:timing-stability:norm-bound:generic}).
	
	\subsubsection{Interval Computation of the Spectral Norm} By \cite[p. 5]{Rump2010}, an upper bound for the spectral norm of matrices with small entries can be efficiently computed by
	\begin{equation}
	\|A\|_{\sigma} \leq \sqrt{\sum_{i,j} a_{i,j}^2} \quad \forall A = (a_{i,j})_{i,j} \in  \R^{n \times n}\label{eq:timing-stability:rough-spectral-norm-bound}.
	\end{equation}
	
	For general matrices, relatively precise bounds for the spectral norm can be determined from the singular value decomposition $\Sigma=U^\transp A V$, where $\Sigma=\diag(\sigma_1, \dots, \sigma_n)$ is the diagonal matrix of singular values of $A$, $UU^\transp=U^\transp U=I$ and $VV^\transp=V^\transp V=I$ \cite[Theorem 5.6.3 and Fact 3.11.4]{Bernstein2009}:
	
	Let $\tilde V$ be a numerical approximation of $V$ with unknown accuracy. All following computations must be in interval arithmetic and are due to \cite[Theorem 3.2]{Rump2010}. Compute $D + E = \tilde V^\transp A^\transp A \tilde V$, where $D=\diag(d_1, \dots, d_n)$ is the diagonal part and $E$ the rest, to approximate
	\begin{equation}
	V^\transp A^\transp A V = V^\transp A^\transp \underbrace{U U^\transp}_{I} A V = \Sigma^\transp \Sigma = \diag(\sigma_1^2, \dots, \sigma_n^2).
	\end{equation} By \eqref{eq:timing-stability:rough-spectral-norm-bound}, compute $\alpha$ such that $\|I-\tilde V^\transp \tilde V\|_{\sigma} \le \alpha < 1$ and $\epsilon$ such that $\|E\|_{\sigma}<\epsilon$. Then,
	\begin{equation}
	\sqrt{\frac{\max_i d_i - \epsilon}{1+\alpha}} \le \|A\|_{\sigma} \le \sqrt{\frac{\max_i d_i + \epsilon}{1-\alpha}}. \label{eq:timing-stability:good-spectral-norm-bound}
	\end{equation}
	This computation has a complexity of $\landauO(n^3)$ \cite[p. 378]{Rump2010}.
	
	\subsubsection{Interval Computation of the Matrix Exponential}
	The matrices $M_1,M_2$ in \cref{sec:timing-stability:normbound:expm} depend on $\ee^{A\idxContinuous T/2}$ in some cases. Therefore, a validated computation of the matrix exponential is required. This is done using functions provided by the Python \emph{mpmath} library.
	
	This also solves the problem that $A\idxContinuous T/2$ may be not exactly known or not exactly representable by floating point values.
	
	\paragraph{Computational Complexity}
	The exponentiation of interval matrices with specified accuracy is NP-hard \cite{Goldsztejn2009} and therefore any known algorithm is of worse than polynomial complexity. As the dimension of $A$ is $n=n\idxPlant+n\idxDiscrete+m+p$, this suggests that an increase in the number $m+p$ of sensors and actuators leads to an exponentially (or worse than polynomially) increasing amount of computation time, effectively invalidating the advantage stated in \cref{remark:timing-stability:complexity} (Increasing $m+p$ requires only a polynomially increasing number of norm bounds).
	
	However, this is not true, as the structure \eqref{eq:timing-stability:splitting:lis-control-model:exp-a-tau} of $\ee^{A\idxContinuous \tau}$ reveals that only the terms $\ee^{A\idxPlant\tau}$ and $\int_0^\tau \ee^{A\idxPlant \xi} d\xi B\idxPlant$ need to be computed. For constant $A\idxPlant$, all terms except $B\idxPlant$ are fixed, so that increasing $m$ only incurs the polynomial complexity of matrix multiplication, and $p$ is irrelevant for this step.
%
}

\section{Synthesis of $P$ via \acsp{LMI}}\label{sec:lmi}
To show stability using \cref{theorem:timing-stability:norm-bound}, the \ac{CQLF} matrix $P$ must be determined such that the bound $\tilde \rho$ is less than 1:
\begin{multline}
\|A_k\|_P \le \|A(\Delta t=0)\|_P + \sum \|\Delta A_{\subsActuate,\dots}\|_P \ifpaper{\\} +  \sum \|\Delta A_{\subsMeasure,\dots}\|_P + \sum \|\Delta A_{\subsActuate\subsMeasure,\dots}\|_P \le \tilde \rho. \ifpaper{\!\!} \label{eq:timing-stability:lmi:goal}
\end{multline}
\Cref{theorem:timing-stability:extreme-lyapunov-lti} guarantees the existence of $P$ with $\|A(\Delta t=0)\|_P < 1$. Because the resulting bounds for $\|\Delta A_{\dots}\|_P$ are often prohibitively large, remaining degrees of freedom in $P$ must be used to minimize $\tilde \rho$ and show stability by $\tilde \rho < 1$. For this we employ an \ac{LMI}-based approach.

\subsection{Validity of Approximations}
\label{sec:timing-stability:lmi:approx-is-fine}
As shown in the following, determining $P$ using \acp{LMI} entails finite numerical precision and approximations. It is important to note that the final stability result is valid no matter how $P$ was determined, as long as $P \succ 0$: The underlying theorems are valid for any $P$-ellipsoidal norm $\|\cdot\|_P$ with $P \succ 0$. 
In the implementation, the numerical result $P$ is checked for $P \succ 0$ and \cref{theorem:timing-stability:norm-bound} using interval arithmetic and the results of \cref{sec:timing-stability:normbound}. If these tests succeed, the system is stable. Otherwise, no conclusion can be drawn.
\ifreport{\xspace(In the implementation, the condition $P\succ 0$ is implicitly checked during the computation of $(P^{1/2})^{-1}$.)\xspace}

\newcommand{\PNormLMIDerivation}{%
	\begin{align}
	\|M\|_P < c \quad \ifreportX{&\Leftrightarrow \quad \max_{x \in \R^n} \sqrt{\frac{(Mx)^\transp P (Mx)}{x^\transp P x}} < c \\
		&\Leftrightarrow  \quad \sqrt{\frac{(Mx)^\transp P (Mx)}{x^\transp P x}} \prec c \\
		&\Leftrightarrow \quad x^\transp M^\transp P M x \prec x^\transp P c^2 x\\}
	&\Leftrightarrow \quad M^\transp P M \prec c^2 P\ifreportX{.}\label{eq:timing-stability:p-norm-to-lmi}
	\end{align}%
}
\ifreportStartMark
\ifreportX{
\subsection{LMI Equivalence of Norm Bounds} \label{sec:lmi:basic-norm-bounds}
The minimum or maximum eigenvalue $\lambda_{\{\min,\max\}}$ can be formulated as \ac{LMI} \cite[Lemma 8.4.1]{Bernstein2009} via
\begin{align}
	\lambda_{\min}(M) > c &\Leftrightarrow  M \succ c I\ifreportX{\\}\ifpaper{,&\!}
		\lambda_{\max}(M) < c &\Leftrightarrow  M \prec c I.\ifpaper{\!}
\end{align}The same is possible for the singular values $\sigma_{\{\min,\max\}}(M) = \lambda^{1/2}_{\{\min,\max\}}(M^\transp M)$:
\begin{align}
\|M\|_{\sigma} = \sigma_{\max}(M) < c \quad &\Leftrightarrow \quad  	M^\transp M \prec c^2 I,\\
	\sigma_{\min}(M) > c \quad &\Leftrightarrow \quad M^\transp M \succ c^2 I.
\end{align}
%
%
A similar result for the $P$-ellipsoid norm can be derived from its definition and the definition of $\succ$ (cf. \cref{def:positive-definite}):
\PNormLMIDerivation
}
\ifreportEndMark

\subsection{\ac{LMI} Problem Formulation}
\label{sec:timing-stability:lmi:lmi}
To use the efficient framework of \acp{LMI}, the $P$-ellipsoid norms in \eqref{eq:timing-stability:lmi:goal} can be expressed \ifreport{using \eqref{eq:timing-stability:p-norm-to-lmi} }\ifpaper{using \PNormLMIDerivation \specificReferenceToReport{sec:lmi:basic-norm-bounds} }as
\begin{align}
A^\transp P \,A &\prec  \bar \rho^{2} P \quad (\Leftrightarrow \|A\|_P < \bar\rho), \label{eq:timing-stability:lmi:pnorm-a}\\
\Delta A_i^\transp P \,\Delta A_i &\prec \beta^{2} P \quad (\Leftrightarrow \|\Delta A_i\|_P < \beta) \quad\ifpaper{\!\!} \forall \Delta A_i \in \mathcal{D}, \label{eq:timing-stability:lmi:pnorm-delta}
\end{align}
where $A=A(\Delta t=0)$ is the nominal-case dynamics and, for now, $\mathcal{D}$ the set of $\Delta A_{\dots}$ in \cref{theorem:timing-stability:splitting} for all possible $\Delta t$.
Ignoring numerical errors, this leads to
\begin{align}
	\ifpaper{\textstyle}
	\|A\|_P + \sum_{\dots} \|\Delta A_{\dots}\|_P \stackrel{\text{\eqref{eq:timing-stability:lmi:pnorm-a}, \eqref{eq:timing-stability:lmi:pnorm-delta}}}{<} \bar\rho + \sum_{\dots} \beta \label{eq:timing-stability:lmi:rho-tilde-less}
\end{align}
and the optimization goal
\begin{align}
	\min_{P\succ 0, \,\bar\rho>0, \,\beta>0} \left(\bar\rho + \ifpaper{\textstyle}\sum_{\dots} \beta\right) \text{  \quad \ifpaper{\!\!\!\!}subject to  \labelcref{eq:timing-stability:lmi:pnorm-a,eq:timing-stability:lmi:pnorm-delta}}.
\end{align}
However, this is not a valid \ac{LMI} because \eqref{eq:timing-stability:lmi:pnorm-a} contains a product of the optimization variables $P$ and $\bar\rho$. Additionally, to avoid numerically ill-conditioned $P$, the constraint
\begin{align}
\gamma I \prec P  &\prec I\quad  (\Leftrightarrow \lambda_{\min}(P) > \gamma ~ \land ~ \lambda_{\max}(P) < 1) \label{eq:timing-stability:lmi:robust}
\end{align}
with $\gamma>0$ is added.%
\ifreport{\xspace(Note that scaling $P$ does not affect $\|\cdot\|_P$. Therefore, the absolute value of the upper bound for $\lambda_{\max}(P)$ does not matter, so it is arbitrarily fixed as $\lambda_{\max}(P) < 1$.)

}
The optimization then becomes
\begin{align}
	\max_{P \in \R^{n \times n}, \gamma > 0} \gamma \text{  \quad subject to  \labelcref{eq:timing-stability:lmi:pnorm-a,eq:timing-stability:lmi:pnorm-delta,eq:timing-stability:lmi:robust}},
\end{align}
where the desired norm bounds $\bar \rho$ and $\beta$ are constant within the LMI and instead optimized in an outer loop. Numerical robustness is further improved by preconditioning as detailed later in \ifpaper{\specificReferenceToReport{sec:timing-stability:lmi:precond}}\ifreport{\cref{sec:timing-stability:lmi:precond}}.
\todoOptional{Citation: Bilinear matrix inequalities are significantly more costly to solve.}

While in theory, $\mathcal{D}$ should be the set of all $\Delta A_{\{\subsActuate,\subsMeasure,\subsActuate\subsMeasure\},\dots}$ for a representative set of timings, this would be prohibitively large for systems with many sensors and actuators. It is instead approximated as the set
\begin{multline}
	\mathcal{D} = \big\{ A(\Delta t) - A(0) ~\big|~ \Delta t = \vecSmallT{\Delta t_{\subsActuate}^\transp & \Delta t_{\subsMeasure}^\transp} \in
	\ifpaper{\\ }
	\big(\{ \underline{\Delta t}_{\subsActuate}, 0, \overline{\Delta t}_{\subsActuate} \}
	 \times 
	\{ \underline{\Delta t}_{\subsMeasure}, 0, \overline{\Delta t}_{\subsMeasure}  \} \big)  \setminus \{0\} \big\}
\end{multline}
representing eight extreme combinations of $\Delta t_{\subsActuate}$ and $\Delta t_{\subsMeasure}$. As noted in \cref{sec:timing-stability:lmi:approx-is-fine}, this approximation does not restrict the validity of the final result.

\subsection{Optimization of $\bar\rho$ and $\beta$}
\label{sec:timing-stability:lmi:rho-beta-search}
In the previous \acp{LMI}, the parameters $\bar\rho$ and $\beta$ must be given, whereas the actual goal is to minimize the analysis result $\tilde \rho$. 
Mainly, $\bar\rho$ and $\beta$ should be minimized because, by \eqref{eq:timing-stability:lmi:goal} and \eqref{eq:timing-stability:lmi:rho-tilde-less}, neglecting the approximation of $\mathcal D$,
\begin{equation}
\tilde \rho = \bar \rho + \beta (m+p+mp)
\end{equation}
is a worst-case bound for $\tilde \rho$.
However, there are limits: Experiments show that smaller $\bar \rho$ increases $\|\Delta A_i\|_P$. Because $\beta > \|\Delta A_i\|_P $, $\bar \rho$ should not be too small.
To show stability, $\bar \rho<1$ is desirable. As $\bar \rho > \|A\|_P > \spectralRadius{A}$, we have $\spectralRadius{A} < \bar \rho < 1$. The implementation uses a fixed value $\bar\rho = 0.8 + 0.2 \spectralRadius{A}$ in this range, and a heuristic search for $\beta$\ifpaper{, as detailed in \specificReferenceToReport{sec:timing-stability:lmi:rho-beta-search}.}\ifreport{:}

\ifreport{
\begin{enumerate}
	\item Initially, $\beta = \frac{1}{4} \frac{1-\bar\rho}{m+p+mp}$ and $\delta=2$, where $\delta$ will be explained later.
	\item Repeat the following three times:
	\begin{itemize}
		\item Compute $P$ and $\tilde \rho$
		\item In the exceptional case of $\|A\|_P>1$, the system is probably unstable. Then, retry with smaller $\beta$ (or exit with error).
		\item If $\gamma<10^{-5}$, update $\delta \coloneq 0.45\delta$.
		\item Update $\beta \coloneq \delta \beta \frac{1 - \|A\|_P}{\tilde \rho - \|A\|_P}$.
	\end{itemize}
	\item Return the lowest $\tilde \rho$ found and the corresponding $P$.
\end{enumerate}
For $\delta=1$ and $\|\Delta A_{\dots}\|_P$ proportional to $\beta$, this would converge to $\tilde \rho=1$ at the second iteration. A larger value of $\delta$ potentially achieves lower $\tilde\rho$ at the cost of lower robustness $\gamma$. Experiments suggest that it also helps to speed up convergence.}

\ifreport{
	\subsection{LMI Preconditioning}\label{sec:timing-stability:lmi:precond}
	To improve speed and accuracy of the LMI solver, a state transformation $\tilde A = R^{-1} A R$ and $\tilde {\mathcal D} = \{R^{-1}DR | D \in \mathcal{D}\}$ is applied.
	By the definition of previous LMI, the ideal robustness $\gamma=1$ would be achieved with  $\tilde P=I$. Assuming $\Delta A_i \approx 0$ and $\bar\rho \approx 1$, $\tilde P=I$ is a solution if
	\begin{equation}
	\|\tilde A\|_{\sigma} \stackrel{\text{ \cref{theorem:timing-stability:p-norm-via-spectral}}}{=} \|\tilde A\|_{P=I} \stackrel{\text{LMI}}{<} \bar \rho \approx 1.
	\end{equation}
	Therefore, $R$ should be chosen such that $\|\tilde A\|_{\sigma}< 1$.
	
	\newcommand{\precOrSucc}{\stackrel{(\prec)}{\succ}}
	A lemma required for the following derivation is that $A \prec 0 \Leftrightarrow M^{-\transp} A M^{-1} \prec  0$  for any invertible $M$, as
	\begin{multline}
	A \prec 0 \Leftrightarrow x^\transp A \underbrace{x}_{\coloneq M^{-1} z} \prec 0 \ifpaper{\\} \Leftrightarrow z^\transp M^{-\transp} A M^{-1} z \prec 0 \Leftrightarrow M^{-\transp} A M^{-1} \prec 0. \label{eq:timing-stability:lmi:pd-rewriting}
	\end{multline}
	
	The computations of \cref{sec:timing-stability:lmi:lmi} are denoted as $P_{\mathrm{LMI}}(A, \mathcal{D}, \bar\rho, \beta)$. For improved accuracy, this original LMI is reused as follows:
	\begin{enumerate}
		\item Compute a quadratic Lyapunov function for the nominal case:  $P_{\mathrm{nominal}} = P_{\mathrm{LMI}}(A=\tilde A, \mathcal{D} = \emptyset, \bar\rho=1, \beta=0)$, therefore $\|A\|_{P_{\mathrm{nominal}}} < 1$ (in practice: $\approx 1$).
		\item Choose $R^{-1}=(P_{\mathrm{nominal}}^{1/2})^{\transp}$, which is nonsingular due to $P_{\mathrm{nominal}} \succ 0$ and \cref{theorem:timing-stability:cholesky}. Then, $\|\tilde A\|_\sigma < 1$, as
		\begin{equation}
		\|A\|_{P_{\mathrm{nominal}}}  \stackrel{\text{ \cref{theorem:timing-stability:p-norm-via-spectral}}}{=} \Big\|\underbrace{(P_{\mathrm{nominal}}^{1/2})^{\transp}}_{R^{-1}} A  \underbrace{(P_{\mathrm{nominal}}^{1/2})^{-\transp}}_{R}\Big\|_{\sigma} = \|\tilde A\|_{\sigma}.
		\end{equation}
		\item Compute $\tilde P = P_{\mathrm{LMI}}(A=\tilde A, \mathcal{D}= \tilde \Delta, \bar\rho, \beta)$
		\item Inverse transform $P = R^{-\transp} \tilde P R^{-1}$  due to
		\begin{align}
		\text{\cref{sec:timing-stability:lmi:lmi}} \quad \Rightarrow& \quad  \tilde A^\transp \tilde P \,\tilde A \prec  \bar \rho^{2} \tilde P\\
		\Leftrightarrow& \quad  R^\transp A^\transp R^{-\transp} \tilde P R^{-1} A R \prec \bar \rho^2 \tilde P\\
		\stackrel{\mathclap{\text{\eqref{eq:timing-stability:lmi:pd-rewriting}}}}{\Leftrightarrow}& \quad  A^\transp \underbrace{R^{-\transp} \tilde P R^{-1}}_{\coloneq P} A \prec \bar \rho^2 R^{-\transp} \tilde P R^{-1}\\
		\Leftrightarrow& \quad  A^\transp P A \prec \rho^2 P.
		\end{align}
		This derivation shows that the norm bounds concerning $\bar \rho$ and analogously also $\beta$ hold unchanged.
		
		As the computation never uses $P$, but only $P^{1/2}$, it is desirable to derive an inverse transform for the Cholesky decomposition.
		
		\emph{Proposition:} This inverse transform is $(P^{1/2})^{\transp} = (\tilde P^{1/2})^{\transp} R^{-1}$.
		
		\emph{Proof:}  The statement is true because the Cholesky definition is unique (\cref{theorem:timing-stability:cholesky}) and the proposed value of $P^{1/2}$ fulfills all three conditions of the definition of the Cholesky decomposition:
		\begin{enumerate}
			\item  $P\succ 0$ due to \eqref{eq:timing-stability:lmi:pd-rewriting} and $\tilde P \succ 0$.
			\item $P^{1/2}$ fulfills $P^{1/2} (P^{1/2})^\transp=P$, so it is either the Cholesky decomposition or a transformed (\eg, transposed) variant.
			\item $(P^{1/2})^\transp$ is upper triangular with positive diagonal entries (UT\textsubscript{+}) because:
			\begin{itemize}
				\item It is the product of  UT\textsubscript{+} matrices: $(\tilde P^{1/2})^\transp$ and $R^{-1} = (P_{\mathrm{nominal}}^{1/2})^\transp$ are UT\textsubscript{+} by definition of the Cholesky decomposition.
				\item  The product of two UT\textsubscript{+} matrices is UT\textsubscript{+} \citep[Fact 3.23.12ii]{Bernstein2009}.  
			\end{itemize} 
		\end{enumerate}
	\end{enumerate}
}

\ifreportX{\clearpage}
\section{Experimental Results} \label{sec:experiments}
The approach was prototypically implemented in Python using \emph{CVXPY} for \acp{LMI} and \emph{mpmath} for interval arithmetic. (Open-source code is available at \url{https://github.com/qronos-project/timing-stability-lmi/}.) Stability could successfully be proven for examples C2 and D2 from \cite{Gaukler2019}, for which no previous stability result is known. These examples are the one- (C2) and three-axis (D2) angular rate control of a linearized quadcopter with a period of $T=10\,\text{ms}$ and a timing uncertainty of $\pm$1\,\%. Example D2 is a multivariable system with $m=4$, $p=3$ and a total dimension of $n=16$.

\Cref{table:results} compares the results and computation times obtained using interval arithmetic ($\tilde \rho$, $t$) with those from a simplified approximation ($\tilde \rho_{\mathrm{approx}}$, $t_{\mathrm{approx}}$), in which the norm bounds from \cref{sec:timing-stability:normbound} are replaced by the floating-point maximum $\max_{\tau} \|\Delta A_{\dots}(\tau)\|$ over 100 samples of $\tau$. While this approximation is not guaranteed to be correct, it is about eight times faster. The small deviations $|\tilde \rho_{\mathrm{approx}} - \tilde \rho|$ show that the norm bounds are accurate.

While stability ($\tilde \rho < 1$) can be shown for example D2, this does not hold for doubled timing uncertainty (D2\textsubscript{b}), which may be due to conservatism or due to actual instability. To analyze the scalability, the dimension of D2 was doubled by block-diagonal repetition. By construction, the resulting system D2\textsubscript{c} of dimension $n=32$ has the same stability properties as D2. It can still be analyzed approximately within six minutes and verified within one hour, however at the cost of increased conservatism: Stability can only be shown for reduced timing uncertainty (D2\textsubscript{d}, $\overline{\Delta t_{\subsMeasure}}$ reduced to 1/10th). This conservatism relates to the fact that the summands of \cref{theorem:timing-stability:splitting} are norm-bounded individually, while their total effect is generally less severe.

{\ifpaper{\small }
\begin{table}
	\begin{tabular}{l|rrrrr}\hline
		name & $n$ & $\tilde \rho_{\mathrm{approx}}$ & $|\tilde \rho - \tilde \rho_{\mathrm{approx}}|$ & $t_{\mathrm{approx}}$ & $t$\\ \hline
		C2 & 5 & $0.914$ & $9.2 \cdot 10^{-8}$ & $1.0$ & $1.6$\\
		D2 & 16 & $0.926$ & $9.3 \cdot 10^{-8}$ & $17.5$ & $98.1$\\
		D2\textsubscript{b}: $2\Delta t$ & 16 & $1.073$ & --- & $12.8$ & ---\\
		D2\textsubscript{c}: $2n$ & 32 & $1.021$ & --- & $312.1$ & ---\\
		D2\textsubscript{d}: $2n$, $\frac{\overline{\Delta t}_{\subsMeasure}}{10}$\ifpaper{\!\!} & 32 & $0.979$ & $9.8 \cdot 10^{-8}$ & $308.1$ & $2196.3$\\\hline
	\end{tabular}
\\
	All values are rounded up to the last shown digit. Times are wall-times in seconds on an Intel i7-8750H CPU with 16GB RAM.
	
	$n=n\idxPlant+n\idxDiscrete+m+p$: Total state dimension
	
	$\tilde \rho$: Upper bound on stability factor with interval arithmetic
	
	$\tilde \rho_{\mathrm{approx}}$: Fast approximation of $\tilde \rho$ 
	
	$t_{\mathrm{approx}}, t$: Time for computing $\tilde \rho_{\mathrm{approx}}$, $\rho$.
	
	Modified system parameters are indicated as $2n$ (dimension doubled by repetition) and $K\Delta t$ (timing variable(s) increased by factor $K$)
	\caption{Experimental results}
	\label{table:results}
\end{table}
}

\section{Conclusion}
We presented a stability verification approach for control systems with multiple inputs and outputs under uncertain timing for sensing and actuating. Here, the challenge is that the system dynamics depends on the combination of all individual timing variables, that is, varying jitter for each sensor and actuator.  To avoid the resulting curse of dimensionality, we exploit the system model's structural properties: A decomposition of the discrete-time dynamics leads to summands with at most two timing variables. Subsequently, we can bound these summands in terms of a norm that corresponds to a \acl{CQLF}. The experimental results show that our approach facilitates the stability analysis for moderately complex systems for which, to the best of our knowledge, previously no analysis methods were known.

Future research will be concerned with extending the approach to the nonlinear case and improving the scalability by a more efficient implementation.

\appendix
\ifreportStartMark

\ifreportEndMark
\printbibliography
\end{document}